\documentclass[11pt,a4paper,twoside]{amsart}
\usepackage[activeacute,english]{babel}
\usepackage[latin1]{inputenc}
\usepackage{amsfonts}
\usepackage{amsmath,amsthm}
\usepackage{amssymb}
\usepackage{graphics}
\usepackage{yfonts}
\usepackage{ upgreek }

\newtheorem{theorem}{Theorem}[section]
\newtheorem{lemma}[theorem]{Lemma}
\newtheorem{proposition}[theorem]{Proposition}
\newtheorem{corollary}[theorem]{Corollary}
{\theoremstyle{remark} \newtheorem{remark}[theorem]{Remark}}

\newcommand{\CC}{{\mathcal C}}

\newcommand{\FF}{{\mathcal F}}

\newcommand{\XX}{{\mathcal X}}

\newcommand{\C}{{\mathbb C}}

\newcommand{\R}{{\mathbb R}}

\newcommand{\real}{{\operatorname{Re}}}

\newcommand{\dist}{{\operatorname{dist}}}
\newcommand{\Tr}{{\operatorname{t}}}
\newcommand{\Ran}{{\operatorname{rn}}}
\newcommand{\Ker}{{\operatorname{kr}}}
\newcommand{\Capa}{{\operatorname{Cap}}}
\newcommand{\Area}{{\operatorname{Area}}}
\newcommand{\Volume}{{\operatorname{Vol}}}

\title[Isoperimetric-type inequality for electrostatic shell interactions]{An isoperimetric-type inequality for electrostatic shell interactions for Dirac operators}
\author[N. Arrizabalaga, A. Mas, L. Vega]{Naiara Arrizabalaga, Albert Mas, Luis Vega}
\date{}
\subjclass[2010]{Primary 81Q10, Secondary 35Q40.}
\keywords{Dirac operator, shell interaction, singular integral.}
\thanks{Arrizabalaga was supported in part by MTM2011-24054 and IT641-13.
Mas was supported by the {\em Juan de la Cierva} program JCI2012-14073 (MEC, Gobierno de Espa\~na), ERC grant 320501 of the European Research Council (FP7/2007-2013), MTM2011-27739 and  MTM2010-16232 (MICINN, Gobierno de Espa\~na), and IT-641-13 (DEUI, Gobierno Vasco). Vega was partially supported by SEV-2013-0323, MTM2011-24054 and IT641-13.}
\address{N. Arrizabalaga. Departamento de Matem\'aticas, Universidad del Pa\'is Vasco/Euskal Herriko Unibertsitatea, UPV/EHU, 48080 Bilbao (Spain)}
\email{naiara.arrizabalaga@ehu.es}
\address{A. Mas. Departament de Matem\`atica Aplicada I,
ETSEIB, Universitat Polit\`ecnica de Catalunya, Avda. Diagonal 647, 08028 Barcelona (Spain)}
\email{amasblesa@gmail.com}
\address{L. Vega. Departamento de Matem\'aticas, Universidad del Pa\'is Vasco/Euskal Herriko Unibertsitatea, UPV/EHU, 48080 Bilbao and BCAM, Alameda de Mazarredo, 14
E-48009 Bilbao (Spain)}
\email{luis.vega@ehu.es, lvega@bcamath.org}

\begin{document}

\begin{abstract}
In this article we investigate spectral properties of the coupling $H+V_\lambda$, where $H=-i\alpha\cdot\nabla
+m\beta$ is the free Dirac operator in $\R^3$, $m>0$ and $V_\lambda$ is an electrostatic shell potential (which depends on a parameter $\lambda\in\R$) located on the boundary of a smooth domain in $\R^3$. Our main result is an isoperimetric-type inequality for the admissible range of $\lambda$'s for which the coupling $H+V_\lambda$ generates pure point spectrum in $(-m,m)$. That the ball is the unique optimizer of this inequality is also shown. Regarding some ingredients of the proof, we make use of the Birman-Schwinger principle adapted to our setting in order to prove some monotonicity property of the admissible $\lambda$'s, and we use this to relate the endpoints of the admissible range of $\lambda$'s to the sharp constant of a quadratic form inequality, from which the isoperimetric-type inequality is derived.
\end{abstract}

\maketitle

\section{Introduction}
We investigate spectral properties of operators that are obtained as the coupling of the free Dirac operator in $\R^3$ with singular measure-valued potentials.
Given $m\geq0$, the free Dirac operator in $\R^3$ is defined by
$H=-i\alpha\cdot\nabla+m\beta,$
where $\alpha=(\alpha_1,\alpha_2,\alpha_3)$,
\begin{equation}\label{freedirac}
\begin{split}
\alpha_j
=&\left(\begin{array}{cc} 0 & \sigma_j\\
\sigma_j & 0 \end{array}\right)
\quad\text{for }j=1,2,3,\quad
\beta
=\left(\begin{array}{cc} I_2 & 0\\
0 & -I_2 \end{array}\right),\quad
I_2
=\left(\begin{array}{cc} 1 & 0\\
0 & 1 \end{array}\right),\\
&\text{and}\quad
\sigma_1
=\left(\begin{array}{cc} 0 & 1\\
1 & 0 \end{array}\right),
\quad\sigma_2
=\left(\begin{array}{cc} 0 & -i\\
i & 0 \end{array}\right),
\quad\sigma_3
=\left(\begin{array}{cc} 1 & 0\\
0 & -1 \end{array}\right)
\end{split}
\end{equation}
compose the family of {\em Pauli matrices}. Although one can take $m=0$ in the definition of $H$, throughout this article we always assume $m>0$ to allow the existence of a nontrivial pure point spectrum in the interval $(-m, m)$ for the corresponding couplings.

Following \cite{AMV1,AMV2}, we consider Hamiltonians of the form $H+V$, being $V$ a singular potential located at the boundary of a bounded smooth domain. These type of couplings are usually referred as {\em shell interactions} for $H$. The particular case of the  sphere was studied in \cite{Dittrich}, while in \cite{AMV1,AMV2} we considered boundaries of general bounded smooth domains. Due to the singularity of the potentials under study, a first issue to be treated is the self-adjoint character of the operator, something that we dealt with in  \cite{AMV1}. Our approach fits within the abstract one developed in \cite{Posi1,Posi2}, although we were interested in some concrete potentials that allowed us to obtain more specific results.

This article is addressed to the particular case of electrostatic shell potentials. Let $\Omega\subset\R^3$ be a bounded smooth domain, let $\sigma$ and $N$ be the surface measure and outward unit normal vector field on ${\partial\Omega}$, respectively. For convenience, we also set $\Omega_+=\Omega$ and $\Omega_-=\R^3\setminus\overline{\Omega}$, so ${\partial\Omega}=\partial\Omega_\pm$. Given $\lambda\in\R$ and $\varphi:\R^3\to\C^4$, the electrostatic shell potential $V_\lambda$ applied to $\varphi$ is formally defined as
$$V_\lambda(\varphi)=\frac{\lambda}{2}(\varphi_++\varphi_-)\sigma,$$
where $\varphi_\pm$ denote the boundary values of $\varphi$ (whenever they exist in a reasonable sense) when one approaches to ${\partial\Omega}$ from $\Omega_\pm$. Therefore, $V_\lambda$ maps functions defined in $\R^3$ to vector measures of the form $f\sigma$ with $f:{\partial\Omega}\to\C^4$. In particular, one can interpret $V_\lambda$ as the distribution $\lambda\delta_{\partial\Omega}$ when acting on functions which have a well-defined trace on ${\partial\Omega}$, where $\delta_{{\partial\Omega}}$ denotes the Dirac-delta distribution on ${\partial\Omega}$.

Our interest is focused on the study of the existence of stable energy states in $(-m,m)$ for $H+V_\lambda$, where $m>0$ is interpreted as the mass of the particle whose evolution is modeled by the coupling $\partial_t+i(H+V_\lambda)$. More precisely, we look for a description of the set of $\lambda$'s in $\R$ for which there exist $a\in(-m,m)$  and a nontrivial spinor $\varphi$ in $L^2(\R^3)^4$ (actually, in the domain of definition of $H+V_\lambda$) such that
\begin{equation}\label{1.1}
(H + V_\lambda)(\varphi)=a\varphi.
\end{equation}
In \cite{AMV2} we found that this is not possible if $|\lambda|$ is either too big or too small. More precisely, we showed that there exist upper and lower thresholds $\lambda_u({\partial\Omega})$ and $\lambda_l({\partial\Omega})$, respectively,  with $0<\lambda_l({\partial\Omega})\leq2\leq\lambda_u({\partial\Omega})$ and such that if $|\lambda|\not\in [\lambda_l({\partial\Omega}),\lambda_u({\partial\Omega})]$ then there exists no nontrivial $\varphi$ verifying (\ref{1.1}) for any $a\in(-m,m)$.

The main purpose of this paper is to determine how small can $[\lambda_l({\partial\Omega}),\lambda_u({\partial\Omega})]$ be under some constraint on the size of ${\partial\Omega}$ and/or $\Omega$. %In fact, a simple scaling argument shows that $I({\partial\Omega})$ can be made as small as wanted.
In Section \ref{ss Newton} we show that a natural condition is to consider
$$\frac{\text{Area}({\partial\Omega})}{\text{Cap}(\Omega)}
=\operatorname{constant},$$
where ${\text{Cap}( \Omega)}$ stands for the Newtonian capacity of $\Omega$ (see Section \ref{ss Newton} for the details). In particular, our main result in this direction is the following theorem (see also Remark \ref{rr2}). The symbol ``$\Ker$'' in the statement of the theorem denotes the {\em kernel}, referring to (\ref{1.1}).
\begin{theorem}\label{intro thm 1}
Let $\Omega\subset\R^3$ be a bounded domain with smooth boundary and assume that
\begin{equation}\label{constr q}
m\,\frac{\Area(\partial\Omega)}{\Capa(\Omega)}>\frac{1}{4\sqrt{2}}.
\end{equation}
Then
\begin{equation*}
\begin{split}
&\sup\{|\lambda|:\,\Ker(H+V_{\lambda}-a)\neq0\text{ for some }a\in(-m,m)\}\\
&\qquad\qquad\qquad\qquad\qquad\geq
4\Bigg(m\,\frac{\Area(\partial\Omega)}{\Capa(\Omega)}+\sqrt{m^2\left(\frac{\Area(\partial\Omega)}{\Capa(\Omega)}\right)^2
+\frac{1}{4}}\Bigg),\\
&\inf\{|\lambda|:\,\Ker(H+V_{\lambda}-a)\neq0\text{ for some }a\in(-m,m)\}\\
&\qquad\qquad\qquad\qquad\qquad\leq
4\Bigg(-m\,\frac{\Area(\partial\Omega)}{\Capa(\Omega)}+\sqrt{m^2\left(\frac{\Area(\partial\Omega)}{\Capa(\Omega)}\right)^2
+\frac{1}{4}}\Bigg).
\end{split}
\end{equation*}
In both cases, the equality holds if and only if $\Omega$ is a ball.
\end{theorem}

The first step to prove this result is to use the connection made in  \cite{AMV2} between \eqref{1.1} and the existence of  a nontrivial eigenvalue $c(a)$ of $C^a_\sigma$, a Cauchy type operator defined on ${\partial\Omega}$ in the principal value sense, and whose precise definition we postpone to Section \ref{ss funda}. This connection  corresponds to the so-called Birman-Schwinger principle (see \cite{Seba}) adapted to our setting (see Proposition \ref{spec p1}).

The second step is to show that $c(a)$ is a monotone function of $a\in(-m,m)$. This has important consequences because it reduces the problem to the study of the limiting cases $a=\pm m.$ Using the well-known properties of the Cauchy operator stated in Lemma \ref{l jump} below, it is sufficient to consider just the case $a= m.$ This latter problem is equivalent to find $\lambda\in\R$ and
$u,h\in L^2(\sigma)^2$ with $u,v\neq0$ such that
 \begin{equation*}%\label{quad eq8}
\begin{split}
\quad\left\{
\begin{array}{rr}  2mK(u)+W(h)=-u/\lambda,  \\  W(u)=-h/\lambda,  \end{array}\right.
\end{split}
\end{equation*}
where $K$ is an operator on $\partial\Omega$ defined by the convolution with the Newtonian kernel $k(x)=(4\pi|x|)^{-1}$ (a positive and compact operator), and $W$ is a ``Clifford algebra'' version of the 2-dimensional Riesz transform on ${\partial\Omega}$ whose precise definition we postpone to Section \ref{s quadratic}.

At this point two results become crucial. On one hand, we use that $2W$ is an isometry when ${\partial\Omega}$ is a sphere. This is indeed something specific of the sphere, in \cite{HMMPT} the authors prove that the spheres are the only boundaries of bounded domains for which $2W$ is an isometry (under some extra assumptions). On the other hand, to deal with $K$, we use the fact proved in \cite{Reichel1,Reichel2} which says that if the Newtonian capacity $\text{Cap}(\Omega)$ is attained on the normalized surface measure of ${\partial\Omega}$ and $\Omega$ is regular enough, then ${\partial\Omega}$ is a sphere. By a simple argument, we relate $K$ and $\text{Cap}(\Omega)$. In order to use these two ingredients, we first prove that to solve our optimization problem is equivalent to minimize, in terms of $\Omega$, the infimum over all $\lambda>0$ such that
\begin{equation}\label{quadratic 1}
\bigg(\frac{4}{\lambda}\bigg)^2\!\!\int_{\partial\Omega}|W(f)|^2\,d\sigma
+\frac{8m}{\lambda}\int_{\partial\Omega} K(f)\cdot\overline{f}\,d\sigma
\leq\int_{\partial\Omega}|f|^2\,d\sigma
\end{equation}
for all $f\in L^2(\sigma)^2$. It is to this infimum $\lambda$ to which we prove an isoperimetric-type inequality like the first one in Theorem \ref{intro thm 1} (see Lemma \ref{iso l3}). The constraint (\ref{constr q}) appears as a technical obstruction on the arguments that we use to connect the infimum $\lambda$ of the quadratic form inequality to the admissible $\lambda$'s that generate eigenvalues as in (\ref{1.1}) (see Theorem \ref{quad t1}$(iv)$ and Corollary \ref{quad c1}, see also Remark \ref{constraint remark} for a related result). 
We should mention that the free Dirac operator $H$ is not bounded neither above nor below, so that characterizing eigenvalues by minimizing some appropriately chosen quadratic form is not straightforward as can be seen in [4]. Since $W$ is self-adjoint, (\ref{quadratic 1}) can be read as
\begin{equation*}
\int_{\partial\Omega}
\bigg(\bigg(\frac{4W}{\lambda}\bigg)^2
+\frac{8mK}{\lambda}\bigg)(f)\cdot\overline{f}\,d\sigma
\leq\int_{\partial\Omega}|f|^2\,d\sigma.
\end{equation*}

The paper is organized as follows. In Section \ref{s preliminaries} we state the preliminaries, where we introduce some notation and recall some properties of the resolvent of $H$, as well as the construction of $H+V_\lambda$. Section \ref{s monotony} is devoted to the Birman-Schwinger principle and the above-mentioned monotone character of the eigenvalues of $C^a_\sigma$. The relation with the limiting case $a=m$ for the optimization problem and the optimal constant of the quadratic form inequality (\ref{quadratic 1}) is explored in Section \ref{s quadratic}. 
Finally, Section \ref{ss isoper} is about the isoperimetric-type result and contains the proof of Theorem \ref{intro thm 1} in Section \ref{ss Newton}. Previously, some other natural constraint conditions not including Newtonian capacity are discarded.

We want to thank P. Exner for enlightening conversations.

\section{Preliminaries}\label{s preliminaries}
This article continues the study developed in \cite{AMV1,AMV2}, so we assume that the reader is familiar with the notation, methods and results in there. However, in this section we recall some basic rudiments for the sake of completeness.

Given a positive Borel measure $\nu$ in $\R^3$, set
$$L^2(\nu)^4=\left\{f:\R^3\to\C^4\text{ $\nu$-measurable}:\,
\int|f|^2\,d\nu<\infty\right\},$$
and denote by $\langle\cdot,\cdot\rangle_{\nu}$ and $\|\cdot\|_\nu$ the standard scalar product and norm in $L^2(\nu)^4$, i.e.,
$\langle f,g\rangle_\nu=\int f\cdot\overline g\,d\nu$ and
$\|f\|^2_\nu=\int|f|^2\,d\nu$ for $f,g\in L^2(\nu)^4$.  We write $I_4$ or $1$ interchangeably to denote the identity operator on $L^2(\nu)^4$.  We say that $\nu$ is $d$-dimensional if there exists $C>0$ such that $\nu(B(x,r))\leq Cr^d$ for all $x\in\R^3$, $r>0$.

We denote by $\mu$ the Lebesgue measure in $\R^3$. Concerning $\partial\Omega$, note that $\sigma$ is 2-dimensional. Since we are not interested in optimal regularity assumptions, for the sequel we assume that ${\partial\Omega}$ is of class $\CC^\infty$.
Finally, we introduce the auxiliary space of locally finite measures $$\XX=\left\{G\mu+g\sigma:\,G\in L^2(\mu)^4,\, g\in L^2(\sigma)^4\right\}.$$

\subsection{A fundamental solution of $H-a$}\label{ss funda}
Observe that $H$, which is symmetric and initially defined in $\CC^\infty_c(\R^3)^4$ ($\C^4$-valued functions in $\R^3$ which are $\CC^\infty$ and with compact support), can be extended by duality to the space of distributions with respect to the test space $\CC^\infty_c(\R^3)^4$ and, in particular, it can be defined on $\XX$.
The following lemma (see \cite[Lemma 2.1]{AMV2}) is concerned with the resolvent of $H$, which will be very useful for the results below.
\begin{lemma}\label{resolvent}
Given $a\in\R$, a fundamental solution of $H-a$ is given by
$$\phi^a(x)=\frac{e^{-\sqrt{m^2-a^2}|x|}}{4\pi|x|}\left(a+m\beta
+\left(1+\sqrt{m^2-a^2}|x|\right)\,i\alpha\cdot\frac{x}{|x|^2}\right)\quad\text{for }x\in\R^3\setminus\{0\},$$
i.e., $(H-a)\phi^a=\delta_0 I_4$ in the sense of distributions, where $\delta_0$ denotes the Dirac measure centered at the origin. Furthermore, if $a\in(-m,m)$ then $\phi^a$ satisfies
\begin{itemize}
\item[$(i)$] $\phi^a_{j,k}\in\CC^\infty(\R^3\setminus\{0\})$ for all $1\leq j,k\leq 4$,
\item[$(ii)$] $\phi^a(x-y)=\overline{(\phi^a)^t}(y-x)$ for all $x,y\in\R^3$ such that $x\neq y$,
\item[$(iii)$] there exist $\gamma,\delta>0$ such that
\begin{itemize}
\item[$(a)$] $\sup_{1\leq j,k\leq 4}|\phi^a_{j,k}(x)|\leq C|x|^{-2}$ for all $|x|<\delta$,
\item[$(b)$] $\sup_{1\leq j,k\leq 4}|\phi^a_{j,k}(x)|\leq Ce^{-\gamma|x|}$ for all $|x|>1/\delta$,
\item[$(c)$] $\sup_{1\leq j,k\leq 4}\,
\sup_{\xi\in\R^3}(1+|\xi|^2)^{1/2}|\FF(\phi^a_{j,k})(\xi)|<\infty,$ where $\FF$ denotes the Fourier transform in $\R^3$.
\end{itemize}
\end{itemize}
\end{lemma}
In the lemma above we denoted the complex conjugate of the transpose of $\phi^a$ by $\overline{(\phi^a)^t}$, that is,
$$((\phi^a)^t)_{\,j,k}=\phi^a_{k,j}\quad\text{and}\quad (\overline{\phi^a})_{j,k}=\overline{\phi^a_{j,k}}\quad\text{for all }1\leq j,k\leq 4.$$
Note that the assumption $a\in(-m,m)$ in Lemma \ref{resolvent} is only relevant for the validity of properties $(ii)$, $(iii)(b)$ and $(iii)(c)$.

Given a positive Borel measure $\nu$ in $\R^3$, $f\in L^2(\nu)^4$, and $x\in\R^3$, we set $$(\phi^a*f\nu)(x)=\int\phi^a(x-y)f(y)\,d\nu(y),$$ whenever the integral makes sense. By Lemma \ref{resolvent} and \cite[Lemma 2.1]{AMV1},
if $a\in(-m,m)$ and $\nu$ is a $d$-dimensional measure in $\R^3$ with $1<d\leq 3$, then there exists $C>0$ such that
\begin{equation}\label{---1}
\|\phi^a*g\nu\|_{\mu}\leq C\|g\|_{\nu}\quad\text{for all $g\in L^2(\nu)^4$}.
\end{equation}

The next lemma (see \cite[Lemma 2.2]{AMV2}), will be used in the sequel.
\begin{lemma}\label{l jump}
Given $g\in L^2(\sigma)^4$ and $x\in{\partial\Omega}$, set
\begin{equation*}
\begin{split}
C_\sigma^a (g)(x)=\lim_{\epsilon\searrow0}\int_{|x-z|>\epsilon}\phi^a(x-z)g(z)\,d\sigma(z)
\quad\text{and}\quad
C_{\pm}^a(g)(x)=\lim_{\Omega_{\pm}\ni y\stackrel{nt}{\longrightarrow} x}(\phi^a*g\sigma)(y),
\end{split}
\end{equation*}
where $\textstyle{\Omega_{\pm}\ni y\stackrel{nt}{\longrightarrow} x}$ means that $y\in\Omega_{\pm}$ tends to $x\in{\partial\Omega}$ non-tangentially. Then $C_\sigma^a$ and $C_\pm^a$ are bounded linear operators in $L^2(\sigma)^4$. Moreover, the following  holds:
\begin{itemize}
\item[$(i)$] $C_\pm^a =\mp\frac{i}{2}\,(\alpha\cdot N)+C_\sigma^a$ (Plemelj--Sokhotski jump formulae),
\item[$(ii)$] for any $a\in[-m,m]$, $C_\sigma^a$ is self-adjoint and $-4(C_\sigma^a(\alpha\cdot N))^2=I_4$.
\end{itemize}
\end{lemma}

\subsection{On the divergence theorem for $H-a$}
A simple computation involving the divergence theorem shows that
\begin{equation*}
\begin{split}
\int_{\Omega_\pm}\left((\alpha\cdot\nabla)\varphi\cdot\overline{\psi}
+\varphi\cdot\overline{(\alpha\cdot\nabla)\psi}\right)d\mu
=\pm\int_{{\partial\Omega}}(\alpha\cdot N)\varphi\cdot\overline{\psi}\,d\sigma
\end{split}
\end{equation*}
for all $\varphi,\psi\in W^{1,2}(\chi_{\Omega_\pm}\mu)^4$, where $W^{1,2}(\chi_{\Omega_\pm}\mu)^4$ denotes the Sobolev space of $\C^4$-valued functions defined on $\Omega_\pm$ such that all its components have all their derivatives up to first order in $L^2(\chi_{\Omega_\pm}\mu)$. As a consequence, we easily deduce that
\begin{equation}\label{div formula}
\begin{split}
\int_{\Omega_\pm}\left((H-a)\varphi\cdot\overline{\psi}
-\varphi\cdot\overline{(H-a)\psi}\right)d\mu
=\mp i\int_{{\partial\Omega}}(\alpha\cdot N)\varphi\cdot\overline{\psi}\,d\sigma.
\end{split}
\end{equation}

\subsection{The construction of $H+V$ and its domain of definition}
In what follows we use a nonstandard notation, $\Phi^a$, to define the convolution of measures in $\XX$ with the fundamental solution of $H-a$, $\phi^a$. Capital letters, as $F$ or $G$, in the argument of $\Phi^a$ denote elements of $L^2(\mu)^4$, and the lowercase letters, as $f$ or $g$, denote elements in  $L^2(\sigma)^4$. Despite that this notation is nonstandard, it is very convenient in order to shorten the forthcoming computations.

Given $G\mu+g\sigma\in\XX$, we define $$\Phi^a(G+g)=
\phi^a*G\mu+\phi^a*g\sigma.$$
Then (\ref{---1}) shows that
$\|\Phi^a(G+g)\|_\mu\leq C(\|G\|_\mu+\|g\|_\sigma)$ for some constant $C>0$ and all $G\mu+g\sigma\in\XX$, so $\Phi^a(G+g)\in L^2(\mu)^4$. Moreover, following \cite[Section 2.3]{AMV1} one can show that
$(H-a)(\Phi^a(G+g))=G\mu+g\sigma$ in the sense of distributions. This allows us to define a ``generic'' potential $V$ acting on functions $\varphi=\Phi^a(G+g)$ by
\begin{equation*}
V(\varphi)= -g\sigma,
\end{equation*}
so that $(H-a+V)(\varphi)=G\mu$
in the sense of distributions. For simplicity of notation, we write
$(H-a+V)(\varphi)=G\in L^2(\mu)^4$.

In order to construct a domain of definition where $H+V$ is self-adjoint, in \cite{AMV1} we used the trace operator on ${\partial\Omega}$.
For $G\in\CC_c^\infty(\R^3)^4$, one defines the trace operator on ${\partial\Omega}$ by $\Tr_{\partial\Omega}(G)=G\chi_{{\partial\Omega}}$. Then, $\Tr_{\partial\Omega}$ extends to a bounded linear operator $$\Tr_\sigma:W^{1,2}(\mu)^4\to L^2(\sigma)^4$$ (see \cite[Proposition 2.6]{AMV1}, for example).  From Lemma \ref{resolvent}$(iii)(c)$, we have
$$\|\Phi^a(G)\|_{W^{1,2}(\mu)^4}\leq C\|G\|_{\mu}$$ for some $C>0$ and all $G\in L^2(\mu)^4$ (see \cite[Lemma 2.8]{AMV1}), thus we can define $$\Phi_\sigma^a(G)=\Tr_\sigma(\Phi^a(G))=\Tr_\sigma(\phi^a*G\mu)$$ and it satisfies
$\|\Phi_\sigma^a(G)\|_\sigma\leq C\|G\|_\mu$ for all $G\in L^2(\mu)^4$.
In accordance with the notation introduced in \cite{AMV1}, for the case $a=0$, we write $\Phi$, $\Phi_\sigma$, $C_\pm$ and $C_\sigma$ instead of $\Phi^0$, $\Phi_\sigma^0$, $C_\pm^0$ and $C_\sigma^0$, respectively.

Finally, we recall our main tool to construct domains where $H+V$ is self-adjoint, namely \cite[Theorem 2.11]{AMV1}. Actually, the following theorem, which corresponds to \cite[Theorem 2.3]{AMV2}, is a direct application of \cite[Theorem 2.11]{AMV1} to $H+V$, and we state it here in order to make the exposition more self-contained. Given an operator between vector spaces $S:X\to Y$, denote
$$\Ker(S)=\{x\in X:\, S(x)=0\}\quad\text{and}\quad
\Ran(S)=\{S(x)\in Y:\, x\in X\}.$$

\begin{theorem}\label{pre t1}
Let $\Lambda:L^2(\sigma)^4\to L^2(\sigma)^4$ be a bounded operator. Set $$D(T)=\{\Phi(G+g): G\mu+g\sigma\in\XX,\,\Phi_\sigma(G)=\Lambda(g)\}\subset L^2(\mu)^4
\text{ and $T=H+V$ on $D(T)$,}$$ where $V(\varphi)=-g\sigma$ and $(H+V)(\varphi)=G$ for all $\varphi=\Phi(G+g)\in D(T)$. If $\Lambda$ is self-adjoint and $\Ran(\Lambda)$ is closed, then $T:D(T)\to L^2(\mu)^4$ is an essentially self-adjoint operator. In that case, if $\{{\Phi(h)}:\,h\in\Ker(\Lambda)\}$ is closed in $L^2(\mu)^4$, then $T$ is self-adjoint.
\end{theorem}

In particular, if $\Lambda$ is self-adjoint and Fredholm, then the operator $T$ given by Theorem \ref{pre t1} is self-adjoint.

\subsection{Electrostatic shell potentials}
In \cite[Theorem 3.8]{AMV1} we proved that,
if $\lambda\in\R\setminus\{0,\pm2\}$ and $T$ is the operator defined by
\begin{equation*}
\begin{split}
D(T)=\big\{\Phi&(G+g): G\mu+g\sigma\in\XX,\,\Phi_\sigma(G)=\Lambda(g)\big\}\\
&\text{and}\quad T=H+V_\lambda \text{ on } D(T),
\end{split}
\end{equation*}
where
\begin{equation}\label{elec lambda}
\begin{split}
\Lambda=-(1/\lambda+C_\sigma),\qquad V_\lambda(\varphi)=\frac{\lambda}{2}(\varphi_++\varphi_-)\sigma
\end{split}
\end{equation}
and
$\varphi_\pm=\Phi_\sigma(G)+C_\pm (g)$ for $\varphi=\Phi(G+g)\in D(T)$, then $T:D(T)\subset L^2(\mu)^4\to L^2(\mu)^4$ is self-adjoint. Moreover, we also showed that $V_\lambda=V$ on $D(T)$ for all $\lambda\neq0$, so the self-adjointness was a consequence of Theorem \ref{pre t1}.
Let us mention that if one replaces $\Phi_\sigma(G)=\Lambda(g)$ by
$\lambda\Phi_\sigma(G)=\lambda\Lambda(g)$ in the definition of $D(T)$ above, one recovers the well-known fact that $D(H+V_0)=D(H)=W^{1,2}(\mu)^4$ when $\lambda=0$.

\section{Birman-Schwinger principle and monotonicity}\label{s monotony}
We will make use of the following proposition, which corresponds to \cite[Proposition 3.1]{AMV2} and can be understood as the classical Birman-Schwinger principle adapted to our setting.
\begin{proposition}\label{spec p1}
Let $T$ be as in {\em Theorem \ref{pre t1}}. Given $a\in(-m,m)$, there exists $\varphi=\Phi(G+g)\in D(T)$ such that $T(\varphi)=a\varphi$ if and only if $\Lambda(g)=(C_\sigma^a-C_\sigma)(g)$ and $G=a\Phi^a(g)$.
Therefore, $\Ker(T-a)\neq0$ if and only if $\Ker(\Lambda+C_\sigma-C_\sigma^a)\neq0$.
\end{proposition}

The following lemma contains the monotonicity property mentioned in the introduction.
\begin{lemma}\label{mono l1}
Given $a\in[-m,m]$, the eigenvalues of $C_\sigma^a$ form a finite or countable sequence $\emptyset\neq\{c_j(a)\}_j\subset\R$, being $1/4$ the only possible accumulation point of $\{c_j(a)^2\}_j$.
Moreover, $\frac{d}{da}\,c_j(a)>0$ for all $a\in(-m,m)$ and all $j$.

As a consequence, given $a\in(-m,m)$, the set of real $\lambda$'s such
that $\Ker(H+V_{\lambda}-a)\neq0$ form a finite or countable sequence $\emptyset\neq\{\lambda_j(a)\}_j\subset\R$, being $4$ the only possible accumulation point of $\{\lambda_j(a)^2\}_j$. Furthermore, $\lambda_j(a)$ is a strictly monotonous increasing function of $a\in(-m,m)$ for all $j$.
\end{lemma}

\begin{proof}
For any $a\in[-m,m]$, the existence of the sequence $\emptyset\neq\{c_j(a)\}_j\subset\R$ stated in the lemma and its possible accumulation point are guaranteed by \cite[Remark 3.5]{AMV2} (which also holds for $a=\pm m$) and the self-adjointness of $C_\sigma^a$.

Given $a\in[-m,m]$ and $c_j(a)$, let $g_j(a)\in L^2(\sigma)^4$
be such that $\|g_j(a)\|_\sigma=1$ and
\begin{equation}\label{elec eq1}
C_\sigma^a(g_j(a))=c_j(a)g_j(a).
\end{equation}
To differentiate $c_j(a)$ with respect of $a$, we take the scalar product of (\ref{elec eq1}) with $g_j(a)$, so
\begin{equation*}
c_j(a)=\langle c_j(a)g_j(a),g_j(a)\rangle_{\sigma}
=\langle C^a_\sigma (g_j(a)),g_j(a)\rangle_{\sigma}.
\end{equation*}
We abreviate $\partial_a\equiv \frac{d}{da}$. Then, at a formal level,
\begin{equation}\label{elec eq2}
\begin{split}
\partial_ac_j(a)
&=\langle\partial_a\big(C^a_\sigma (g_j(a))\big),g_j(a)\rangle_{\sigma}
+\langle C^a_\sigma (g_j(a)),\partial_ag_j(a)\rangle_{\sigma}\\
&=\langle(\partial_aC^a_\sigma)(g_j(a)),g_j(a)\rangle_{\sigma}
+\langle C^a_\sigma (\partial_ag_j(a)),g_j(a)\rangle_{\sigma}
+\langle C^a_\sigma (g_j(a)),\partial_ag_j(a)\rangle_{\sigma}\\
&=\langle(\partial_aC^a_\sigma)(g_j(a)),g_j(a)\rangle_{\sigma}
+2\real\langle \partial_ag_j(a),C^a_\sigma(g_j(a))\rangle_{\sigma},
\end{split}
\end{equation}
where we used in the last equality above that $C^a_\sigma$ is self-adjoint. Recall that $\|g_j(a)\|_\sigma=1$ for all $a\in(-m,m)$, thus (\ref{elec eq1}) gives
\begin{equation*}
\begin{split}
0&=c_j(a)\partial_a\langle g_j(a),g_j(a)\rangle_{\sigma}
=\langle \partial_a g_j(a),c_j(a)g_j(a)\rangle_{\sigma}
+\langle c_j(a)g_j(a),\partial_a g_j(a)\rangle_{\sigma}\\
&=2\real\langle \partial_a g_j(a),C^a_\sigma(g_j(a))\rangle_{\sigma},
\end{split}
\end{equation*}
which plugged into (\ref{elec eq2}) yields
\begin{equation}\label{elec eq3}
\begin{split}
\partial_ac_j(a)=\langle(\partial_aC^a_\sigma)(g_j(a)),g_j(a)\rangle_{\sigma}.
\end{split}
\end{equation}

To justify the above computations, in particular in what respects to the issue of the principal value in the definition of $C_\sigma^a$, one can decompose the kernel
\begin{equation*}
\begin{split}
\phi^a(x)&=\frac{e^{-\sqrt{m^2-a^2}|x|}}{4\pi|x|}\left(a+m\beta
+i\sqrt{m^2-a^2}\,\alpha\cdot\frac{x}{|x|}\right)+\frac{e^{-\sqrt{m^2-a^2}|x|}-1}{4\pi}\,i\left(\alpha\cdot\frac{x}{|x|^3}\right)\\
&\quad+\frac{i}{4\pi}\,\left(\alpha\cdot\frac{x}{|x|^3}\right)
=:\omega_1(x)+\omega_2(x)+\omega_3(x)
\end{split}
\end{equation*}
and note that the principal value only concerns $\omega_3$, since the kernels $\omega_1$ and $\omega_2$ are absolutely integrable on ${\partial\Omega}$ and actually define compact operators, but $\omega_3$ does not depend on $a$. At this point, standard arguments in perturbation theory (by compact operators which depend continuously on the perturbation parameter) allow us to justify the formal computations carried out above concerning $\partial_a$.

Our aim now is to understand the operator $\partial_aC^a_\sigma$. One may guess that, since $C^a_\sigma$ is defined as the convolution operator on ${\partial\Omega}$ with the fundamental solution of $H-a$,  and formally
$\partial_a((H-a)^{-1})=(H-a)^{-2}$, then $\partial_aC^a_\sigma$ should be defined as the convolution operator on ${\partial\Omega}$ with the fundamental solution of $(H-a)^2$. This is indeed the case. In the following lines, we are going to prove the details of this argument.
We can easily compute
\begin{equation}\label{elec eq4}
\begin{split}
\partial_a(\phi^a(x))=\frac{ae^{-\sqrt{m^2-a^2}|x|}}{4\pi\sqrt{m^2-a^2}}
\left(a+m\beta
+i\sqrt{m^2-a^2}\,\alpha\cdot\frac{x}{|x|}\right)+\frac{e^{-\sqrt{m^2-a^2}|x|}}{4\pi|x|}.
\end{split}
\end{equation}
Note that $$-i\alpha\cdot\nabla(e^{-\sqrt{m^2-a^2}|x|})=
i\sqrt{m^2-a^2}e^{-\sqrt{m^2-a^2}|x|}\,\alpha\cdot\frac{x}{|x|},$$ so (\ref{elec eq4}) gives
\begin{equation}\label{elec eq5}
\begin{split}
\partial_a(\phi^a(x))
=a\left(H+a\right)
\frac{e^{-\sqrt{m^2-a^2}|x|}}{4\pi\sqrt{m^2-a^2}}
+\frac{e^{-\sqrt{m^2-a^2}|x|}}{4\pi|x|}.
\end{split}
\end{equation}
A simple calculation shows that
\begin{equation}\label{elec eq7}
(-\Delta+m^2-a^2)\frac{e^{-\sqrt{m^2-a^2}|x|}}{8\pi\sqrt{m^2-a^2}}
=\frac{e^{-\sqrt{m^2-a^2}|x|}}{4\pi|x|}
\end{equation}
which, combined with (\ref{elec eq5}) and using that
$-\Delta+m^2-a^2
=(H-a)(H+a),$ yields
\begin{equation}\label{elec eq6}
\begin{split}
\partial_a(\phi^a(x))
&=\left(a(H+a)+\frac{1}{2}(-\Delta+m^2-a^2)\right)\frac{e^{-\sqrt{m^2-a^2}|x|}}{4\pi\sqrt{m^2-a^2}}\\
&=\left(a+\frac{1}{2}(H-a)\right)
(H+a)\,\frac{e^{-\sqrt{m^2-a^2}|x|}}
{4\pi\sqrt{m^2-a^2}}
=(H+a)^2\,\frac{e^{-\sqrt{m^2-a^2}|x|}}{8\pi\sqrt{m^2-a^2}}.
\end{split}
\end{equation}
Recall that $(4\pi|x|)^{-1}e^{-\sqrt{m^2-a^2}|x|}$ is a fundamental solution of $-\Delta+m^2-a^2$, that is
$$(-\Delta+m^2-a^2)\frac{e^{-\sqrt{m^2-a^2}|x|}}{4\pi|x|}=\delta_0$$
in the sense of distributions. In particular, from (\ref{elec eq7}) we get that
\begin{equation}\label{elec eq8}
(-\Delta+m^2-a^2)^2\,\frac{e^{-\sqrt{m^2-a^2}|x|}}{8\pi\sqrt{m^2-a^2}}
=\delta_0.
\end{equation}
Since $-\Delta+m^2-a^2$ commutes with
$H+a$, we easily see that
$$(H-a)^2(H+a)^2=
(-\Delta+m^2-a^2)^2,$$
and then, from (\ref{elec eq6}) and (\ref{elec eq8}), we finally deduce that
\begin{equation}\label{elec eq10}
\begin{split}
(H-a)^2\,\partial_a(\phi^a(x))=\delta_0,
\end{split}
\end{equation}
which means that $\partial_a(\phi^a(x))$ is a fundamental solution of
$(H-a)^2$, and $\partial_aC^a_\sigma$ corresponds to the operator of convolution on ${\partial\Omega}$ with this kernel, as suggested before (\ref{elec eq4}). Note that $\partial_a(\phi^a(x))=O(1/|x|)$ for $|x|\to0$, so in particular $\partial_aC^a_\sigma$ is compact in $L^2(\sigma)^4$.

Given $g\in L^2(\sigma)^4$, set
\begin{equation*}
\begin{split}
u(x)&=\int \partial_a(\phi^a(x-y))g(y)\,d\sigma(y)\qquad\text{for }x\in\R^3,
\end{split}
\end{equation*}
so $u=(\partial_aC^a_\sigma)(g)$ on ${\partial\Omega}$.
Using (\ref{elec eq6}), that $-\Delta+m^2-a^2$ and $H+a$ commute and (\ref{elec eq7}), we see that for any $x\in\R^3\setminus{\partial\Omega}$,
\begin{equation}\label{elec eq9}
\begin{split}
(H-a)u(x)
&=\int (H_x-a)
\partial_a(\phi^a(x-y))g(y)\,d\sigma(y)\\
&=\int (-\Delta_x+m^2-a^2)
(H_x+a)\,\frac{e^{-\sqrt{m^2-a^2}|x-y|}}{8\pi\sqrt{m^2-a^2}}\,g(y)\,d\sigma(y)\\
&=\int(H_x+a)\,\frac{e^{-\sqrt{m^2-a^2}|x-y|}}{4\pi|x-y|}\,g(y)\,d\sigma(y)=\Phi^a(g)(x),
\end{split}
\end{equation}
because $\phi^a(x)=(H+a)(4\pi|x|)^{-1}e^{-\sqrt{m^2-a^2}|x|}$ by construction. Concernig the notation employed, we mention that $\Delta_x$ and $H_x$ denote the Laplace and Dirac operators acting as a derivative on the $x$ variable. Since $\phi^a$ is a fundamental solution of $H-a$, we see from (\ref{elec eq9}) that $(H-a)^2u=0$ outside ${\partial\Omega}$,
a fact that we already knew in view of (\ref{elec eq10}) and the definition of $u$.

From Lemma \ref{l jump}$(i)$, we have
$g=i(\alpha\cdot N)(C^a_+(g)-C^a_-(g)).$
Therefore, using (\ref{div formula}), that $(H-a)\Phi^a(g)=0$ outside ${\partial\Omega}$ and (\ref{elec eq9}), we finally get
\begin{equation}\label{elec eq11}
\begin{split}
\langle(\partial_aC^a_\sigma)(g),g\rangle_{\sigma}
&=-i\int (\alpha\cdot N) u\cdot\overline{(C^a_+(g)-C^a_-(g))}\,d\sigma\\
&=\int_{\R^3\setminus{\partial\Omega}}\left((H-a)u\cdot\overline{\Phi^a(g)}
-u\cdot\overline{(H-a)\Phi^a(g)}\right)d\mu\\
&=\int_{\R^3\setminus{\partial\Omega}}|\Phi^a(g)|^2\,d\mu.
\end{split}
\end{equation}
Thanks to the Plemelj--Sokhotski jump formulae from Lemma \ref{l jump}$(i)$, we see that if $g\in L^2(\sigma)^4$ is such that $\Phi^a(g)=0$ in $\R^3\setminus{\partial\Omega}$ then $C_\pm^a(g)=0$, and thus $g=0$. Therefore, applying (\ref{elec eq11}) to $g_j(a)$ and plugging it into (\ref{elec eq3}) yields
\begin{equation*}
\begin{split}
\partial_ac_j(a)&=\langle(\partial_aC^a_\sigma)(g_j(a)),g_j(a)\rangle_{\sigma}=\int_{\R^3\setminus{\partial\Omega}}|\Phi^a(g_j(a))|^2\,d\mu>0,
\end{split}
\end{equation*}
because $g_j(a)$ is not identically zero (since $\|g_j(a)\|_\sigma=1$ by assumption).

To finish the proof of the lemma, it only remains to be shown the stated conclusions about $\{\lambda_j(a)\}_j$.
By Proposition \ref{spec p1} and the definition of $\Lambda$ in (\ref{elec lambda}), if $a\in(-m,m)$ then
$$\text{$\Ker(H+V_{\lambda}-a)\neq0$\quad if and only if\quad$\Ker(1/\lambda+C_\sigma^a)\neq0$,}$$ thus the existence of a sequence $\emptyset\neq\{\lambda_j(a)\}_j\subset\R$ such that $\Ker(H+V_{\lambda_j(a)}-a)\neq0$ and the fact stated in the lemma concerning its unique possible accumulation point follow from the first part of the lemma. Moreover, by setting $c_j(a)=-1/\lambda_j(a)$ we see that $\lambda_j(a)$ is a strictly monotonous increasing function of $a\in(-m,m)$ for all $j$.
\end{proof}

\begin{corollary}\label{elec c1}
Given $a\in(-m,m)$, we have
\begin{equation}\label{elec eq12}
\begin{split}
\sup\{\lambda<0:\,\Ker(H+V_{\lambda}-a)\neq0\}
&=-4/\sup\{\lambda>0:\,\Ker(H+V_{\lambda}-a)\neq0\},
\end{split}
\end{equation}
and the same holds replacing $\sup$ by $\inf$ on both sides of
$(\ref{elec eq12})$.
Set
$$\lambda_{\pm m}^s=\sup\{\lambda\in\R:\,\Ker(1/\lambda+C^{\pm m}_\sigma)\neq0\},\quad\lambda_{\pm m}^i
=\inf\{\lambda\in\R:\,\Ker(1/\lambda+C^{\pm m}_\sigma)\neq0\}.$$
Then $\lambda_{\pm m}^s>0>\lambda_{\pm m}^i$
and the following hold:
\begin{itemize}
\item[$(i)$] $\sup\{\lambda\in\R:\,\Ker(H+V_{\lambda}-a)\neq0\text{ for some }a\in(-m,m)\}=\lambda_{m}^s$,
\item[$(ii)$] $\inf\{\lambda\in\R:\,\Ker(H+V_{\lambda}-a)\neq0\text{ for some }a\in(-m,m)\}=\lambda_{-m}^i$,
\item[$(iii)$] $\sup\{|\lambda|:\,\Ker(H+V_{\lambda}-a)\neq0\text{ for some }a\in(-m,m)\}=\max\{\lambda_{m}^s,-\lambda_{-m}^i\}$,
\item[$(iv)$] $\inf\{|\lambda|:\,\Ker(H+V_{\lambda}-a)\neq0\text{ for some }a\in(-m,m)\}=4/\max\{\lambda_{m}^s,-\lambda_{-m}^i\}$.
\end{itemize}
\end{corollary}

\begin{proof}
Given $a\in(-m,m)$, by \cite[Remark 3.5]{AMV2} we see that  $\{\lambda\in\R\setminus\{0\}:\,\Ker(H+V_{\lambda}-a)\neq0\}$ is a non empty set. Furthermore, in \cite[Theorem 3.3]{AMV2} we also proved that
\begin{equation}\label{elec eq13}
\begin{split}
\Ker(H+V_{\lambda}-a)\neq0\quad\text{if and only if}\quad
\Ker(H+V_{-4/\lambda}-a)\neq0.
\end{split}
\end{equation}
In particular, we see that
$\{\lambda<0:\,\Ker(H+V_{\lambda}-a)\neq0\}$ and
$\{\lambda>0:\,\Ker(H+V_{\lambda}-a)\neq0\}$ are non empty sets. Then, a simple argument using (\ref{elec eq13}) proves (\ref{elec eq12}).

Note that \cite[Remark 3.5]{AMV2} still aplies to the case $a=\pm m$, so
$\{\lambda\in\R:\,\Ker(1/\lambda+C^{\pm m}_\sigma)\neq0\}$ is a non empty set, and thus $\lambda_{\pm m}^s$ and $\lambda_{\pm m}^i$ are well defined. An inspection of the proof of \cite[Theorem 3.3]{AMV2}  shows that, for any $a\in[-m,m]$,
\begin{equation}\label{elec eq14}
\begin{split}
\Ker(1/\lambda+C^{a}_\sigma)\neq0\quad\text{if and only if}\quad
\Ker(-\lambda/4+C^{a}_\sigma)\neq0,
\end{split}
\end{equation}
which in fact is a consequence of Lemma \ref{l jump}$(ii)$ (note that (\ref{elec eq13}) follows by (\ref{elec eq14}) and Proposition \ref{spec p1}). A straightforward application of (\ref{elec eq14}) proves that
$\lambda_{\pm m}^s>0>\lambda_{\pm m}^i$. Furthermore, $(i)$ and $(ii)$ are a direct consequence of the monotonicity property proved in Lemma \ref{mono l1}, and $(iii)$ follows from $(i)$, $(ii)$ and the fact that $\lambda_{m}^s>0>\lambda_{-m}^i$. Regarding $(iv)$, note that $\inf\{|\lambda|:\,\Ker(H+V_{\lambda}-a)\neq0\text{ for some }a\in(-m,m)\}$ is the minimum between
$-\sup\{\lambda<0:\,\Ker(H+V_{\lambda}-a)\neq0\text{ for some }a\in(-m,m)\}$ and
$\inf\{\lambda>0:\,\Ker(H+V_{\lambda}-a)\neq0\text{ for some }a\in(-m,m)\}$, which by (\ref{elec eq12}) and Lemma \ref{mono l1} correspond to $4/\lambda_{m}^s$ and $-4/\lambda_{-m}^i$, respectively. This yields $(iv)$.
\end{proof}

\section{Quadratic forms}\label{s quadratic}
For $a\in\R$ and $\upsigma =(\sigma_1,\sigma_2,\sigma_3)$, where the $\sigma_j$'s compose the family of Pauli matrices introduced in (\ref{freedirac}), define the kernels
\begin{equation*}
k^a(x)=\frac{e^{-\sqrt{m^2-a^2}|x|}}{4\pi|x|}\,I_2
\quad\text{and}\quad
w^a(x)=\frac{e^{-\sqrt{m^2-a^2}|x|}}{4\pi|x|^3}
\left(1+\sqrt{m^2-a^2}|x|\right)\,i\,\upsigma \cdot x
\end{equation*}
for $x\in\R^3\setminus\{0\}$. Given $f\in L^2(\sigma)^2$ and $x\in {\partial\Omega}$, set
\begin{equation*}
K^a(f)(x)=\int k^a(x-z)f(z)\,d\sigma(z)
\quad\text{and}\quad
W^a(f)(x)=\lim_{\epsilon\searrow0}\int_{|x-z|>\epsilon} w^a(x-z)f(z)\,d\sigma(z).
\end{equation*}
That $K^a$ and $W^a$ are bounded operators in $L^2(\sigma)^2$ can be verified similarly to the case of $C^a_\sigma$ in $L^2(\sigma)^4$, we omit the details. Moreover, note that
\begin{equation}\label{sphere eq1}
C_\sigma^a
=\left(\begin{array}{cc}  (a+m)K^a& W^a\\
W^a & (a-m)K^a \end{array}\right).
\end{equation}

The results in the following lemma are contained in \cite[Section 4]{AMV2}.
\begin{lemma}\label{sphere l1}
For any $a\in[-m,m]$, $K^a$ is positive and self-adjoint, $W^a$ is also self-adjoint and the following hold:
\begin{itemize}
\item[$(i)$] the anticommutator $\left\{(\upsigma \cdot N)K^a,(\upsigma \cdot N)W^a\right\}$ vanishes identically,
\item[$(ii)$] $((\upsigma \cdot N)W^a)^2+(a^2-m^2)((\upsigma \cdot N)K^a)^2=-1/4$.
\end{itemize}
\end{lemma}

For simplicity of notation, we write $k$, $w$, $K$ and $W$ instead of
$k^m$, $w^m$, $K^m$ and $W^m$, respectively. Observe that $k(x)=1/(4\pi|x|)$ and $w(x)=i\upsigma\cdot x/(4\pi|x|^3)$.  If $T$ denotes a bounded operator in $L^2(\sigma)^2$, we write
$\|T\|_\sigma$ instead of $\|T\|_{L^2(\sigma)^2\to L^2(\sigma)^2}$.

The following lemma is essentially contained in \cite{HMMPT}, but we give a simple proof for the sake of completeness.
\begin{lemma}\label{quad isometry}
$\|W\|_\sigma\geq1/2$. Moreover, $\|W\|_\sigma=1/2$ if and only if $\{\upsigma\cdot N,W\}=0$, and in this case $2W$ is an isometry in $L^2(\sigma)^2$.
\end{lemma}
\begin{proof}
From Lemma \ref{sphere l1}$(ii)$ we have
\begin{equation}\label{quad isometry eq1}
\begin{split}
\frac{1}{16}\int|f|^2\,d\sigma
=\int|((\upsigma\cdot N)W)^2(f)|^2\,d\sigma
\leq\|W\|_\sigma^2\int|W(f)|^2\,d\sigma
\leq\|W\|_\sigma^4\int|f|^2\,d\sigma
\end{split}
\end{equation}
for all $f\in L^2(\sigma)^2$. From this we see that $\|W\|_\sigma\geq1/2$.

On one hand, if $\|W\|_\sigma=1/2$, then (\ref{quad isometry eq1}) yields
\begin{equation}\label{quad isometry eq2}
\begin{split}
\frac{1}{4}\int|f|^2\,d\sigma
=\int|W(f)|^2\,d\sigma\qquad\text{for all }f\in L^2(\sigma)^2,
\end{split}
\end{equation}
which shows that $2W$ is an isometry in $L^2(\sigma)^2$. By Lemma \ref{sphere l1} and (\ref{quad isometry eq2}) we conclude that
\begin{equation*}
\begin{split}
\int|\{\upsigma\cdot N,W\}(f)|^2\,d\sigma
&=\int ((\upsigma\cdot N)W+W(\upsigma\cdot N))(f)
\cdot\overline{((\upsigma\cdot N)W+W(\upsigma\cdot N))(f)}\,d\sigma\\
&=\int\left(|W(f)|^2+|W(\upsigma\cdot N)(f)|^2
-\frac{1}{4}\,|f|^2-\frac{1}{4}\,|(\upsigma\cdot N)f|^2\right)\,d\sigma=0
\end{split}
\end{equation*}
for all $f\in L^2(\sigma)^2$, which implies that $\{\upsigma\cdot N,W\}=0$.

On the other hand, if $\{\upsigma\cdot N,W\}=0$ then, once again by Lemma \ref{sphere l1},
\begin{equation}\label{quad isometry eq3}
\begin{split}
\int|W(f)|^2\,d\sigma
&=\int W(\upsigma\cdot N)(\upsigma\cdot N)W(f)\cdot\overline{f}\,d\sigma\\
&=\int W(\upsigma\cdot N)\{\upsigma\cdot N,W\}(f)\cdot\overline{f}\,d\sigma+\frac{1}{4}\int|f|^2\,d\sigma
=\frac{1}{4}\int|f|^2\,d\sigma
\end{split}
\end{equation}
for all $f\in L^2(\sigma)^2$. In particular, $\|W\|_\sigma=1/2$.
\end{proof}
We must mention that in \cite{HMMPT} the authors show that $\{\upsigma\cdot N,W\}=0$ (or, equivalently, $\|W\|_\sigma=1/2$) if and only if ${\partial\Omega}$ is a plane or a sphere, as commented in the introduction in reference to the isometric character of $2W$.

The following theorem explores the connection between the quadratic form inequality (\ref{quadratic 1}) and the eigenvalues of
$C_\sigma^{\pm m}$, and it is a key ingredient to derive the isoperimetric-type inequalities contained in Theorem \ref{intro thm 1}.
\begin{theorem}\label{quad t1}
Let ${\lambda_\Omega}$ be the infimum over all $\lambda>0$ such that
\begin{equation}\label{quad eq1}
\bigg(\frac{4}{\lambda}\bigg)^2\!\!\int|W(f)|^2\,d\sigma
+\frac{8m}{\lambda}\int K(f)\cdot\overline{f}\,d\sigma
\leq\int|f|^2\,d\sigma
\end{equation}
for all  $f\in L^2(\sigma)^2$. Then, ${\lambda_\Omega}$ is also the infimum over all $\lambda>0$ such that
\begin{equation}\label{quad eq1''}
\int|f|^2\,d\sigma
+2m\lambda\int K(f)\cdot\overline{f}\,d\sigma
\leq\lambda^2\int|W(f)|^2\,d\sigma
\end{equation}
for all $f\in L^2(\sigma)^2$, and the following hold:
\begin{itemize}
\item[$(i)$] $2<4\big(m\|K\|_\sigma
+\sqrt{m^2\|K\|_\sigma^2+1/4}\big)\leq{\lambda_\Omega}\leq
4\big(m\|K\|_\sigma+\sqrt{m^2\|K\|_\sigma^2+\|W\|_\sigma^2}\big),$
\item[$(ii)$] if $\lambda>0$ is such that $\Ker(1/\lambda+C_\sigma^{m})\neq0$ then
$\lambda\leq{\lambda_\Omega}$,
\item[$(iii)$] if $\lambda<0$ is such that $\Ker(1/\lambda+C_\sigma^{-m})\neq0$ then $\lambda\geq-{\lambda_\Omega}$,
\item[$(iv)$] $(\ref{quad eq1})$ holds for all $\lambda\geq{\lambda_\Omega}$ and it is sharp for $\lambda=\lambda_\Omega$. If $\lambda={\lambda_\Omega}>2\sqrt{2}$ then the equality in $(\ref{quad eq1})$ is attained. In this case, the minimizers of $(\ref{quad eq1})$ (that is, functions that attain the equality) give rise to functions in $\Ker(1/{\lambda_\Omega}+C_\sigma^m)$ and vice versa; the same holds replacing $\Ker(1/{\lambda_\Omega}+C_\sigma^m)$ by $\Ker(-1/{\lambda_\Omega}+C_\sigma^{-m})$,
\item[$(v)$] $(iv)$ also holds replacing $(\ref{quad eq1})$ by $(\ref{quad eq1''})$.
\end{itemize}
\end{theorem}

\begin{proof}
Given $\lambda>0$ and $f\in L^2(\sigma)^2$, set
$$A(\lambda,f)=\bigg(\frac{4}{\lambda}\bigg)^2\!\!\int|W(f)|^2\,d\sigma
+\frac{8m}{\lambda}\int K(f)\cdot\overline{f}\,d\sigma.$$

Let us prove $(i)$. Note that
\begin{equation}\label{quad eq2}
A(\lambda,f)\leq\left(\bigg(\frac{4\|W\|_\sigma}{\lambda}\bigg)^2
+\frac{8m\|K\|_\sigma}{\lambda}\right)\|f\|^2_\sigma.
\end{equation}
Hence, if $\lambda\geq4\big(m\|K\|_\sigma+\sqrt{m^2\|K\|_\sigma^2+\|W\|_\sigma^2}\big)$ then (\ref{quad eq2}) easily yields
$A(\lambda,f)\leq\|f\|_\sigma^2$ for all $f\in L^2(\sigma)^2$, which in turn implies that ${\lambda_\Omega}\leq
4\big(m\|K\|_\sigma+\sqrt{m^2\|K\|_\sigma^2+\|W\|_\sigma^2}\big)$.

The inequality from below is a bit more involved. Let $\lambda>0$ be such that
\begin{equation}\label{quad eq3}
A(\lambda,f)\leq\|f\|_\sigma^2\qquad \text{for all }f\in L^2(\sigma)^2.
\end{equation}
If we set $h=\frac{4}{\lambda}(\upsigma\cdot N)W(f)\in L^2(\sigma)^2$, then  $f=-\lambda(\upsigma\cdot N)W(h)$ by Lemma \ref{sphere l1}$(ii)$ taking $a=m$. Furthermore,
\begin{equation}\label{quad eq4}
\int|W(f)|^2\,d\sigma=\bigg(\frac{\lambda}{4}\bigg)^2\!\!\int|(\upsigma\cdot N)h|^2\,d\sigma=\bigg(\frac{\lambda}{4}\bigg)^2
\!\!\int|h|^2\,d\sigma
\end{equation}
and
\begin{equation}\label{quad eq5}
\int|f|^2\,d\sigma
=\lambda^2\int|(\upsigma\cdot N)W(h)|^2\,d\sigma
=\lambda^2\int|W(h)|^2\,d\sigma.
\end{equation}
Moreover, using Lemma \ref{sphere l1},
\begin{equation}\label{quad eq6}
\begin{split}
\int K(f)\cdot\overline{f}\,d\sigma
&=\lambda^2\int K(\upsigma\cdot N)W(h)\cdot\overline{(\upsigma\cdot N)W(h)}\,d\sigma\\
&=-\lambda^2\int K(\upsigma\cdot N)W(\upsigma\cdot N)W(h)\cdot\overline{h}\,d\sigma
=\frac{\lambda^2}{4}\int K(h)\cdot\overline{h}\,d\sigma.
\end{split}
\end{equation}
Gathering (\ref{quad eq3}) with (\ref{quad eq4}), (\ref{quad eq5}) and (\ref{quad eq6}) yields
\begin{equation}\label{quad eq7}
\int|h|^2\,d\sigma
+2m\lambda\int K(h)\cdot\overline{h}\,d\sigma
\leq\lambda^2\int|W(h)|^2\,d\sigma \qquad \text{for all }h\in L^2(\sigma)^2.
\end{equation}
Note that this argument is reversible, thus in particular we have proven that
\begin{equation*}%\label{quad eq7}
{\lambda_\Omega}=\inf\left\{\lambda>0:\,\int|f|^2\,d\sigma
+2m\lambda\int K(f)\cdot\overline{f}\,d\sigma
\leq\lambda^2\int|W(f)|^2\,d\sigma \quad \forall\, f\in L^2(\sigma)^2\right\},
\end{equation*}
which yields (\ref{quad eq1''}).
If we multiply (\ref{quad eq7}) by $16/\lambda^4$ we get
\begin{equation*}
\frac{16}{\lambda^4}\int|f|^2\,d\sigma
+\frac{32m}{\lambda^3}\int K(f)\cdot\overline{f}\,d\sigma
\leq\frac{16}{\lambda^2}\int|W(f)|^2\,d\sigma\qquad \text{for all }f\in L^2(\sigma)^2,
\end{equation*}
which added to (\ref{quad eq3}) gives
\begin{equation*}
2m\int K(f)\cdot\overline{f}\,d\sigma
\leq\bigg(\frac{\lambda}{4}-\frac{1}{\lambda}\bigg)\int|f|^2\,d\sigma
\qquad \text{for all }f\in L^2(\sigma)^2.
\end{equation*}
Since $K$ is bounded, positive and self-adjoint, we see from the above inequality that
\begin{equation*}
2m\|K\|_\sigma=2m\sup_{\|f\|_\sigma=1}\int K(f)\cdot\overline{f}\,d\sigma
\leq\frac{\lambda}{4}-\frac{1}{\lambda},
\end{equation*}
which in turn is equivalent to
$$\lambda^2-8m\|K\|_\sigma\lambda-4\geq0,$$
since $\lambda>0$ by assumption. Therefore, we must have
$\lambda\geq4\big(m\|K\|_\sigma
+\sqrt{m^2\|K\|_\sigma^2+1/4}\big)$ for all $\lambda>0$ satisfying (\ref{quad eq3}). This gives the desired inequality from below for ${\lambda_\Omega}$, and finishes the proof of $(i)$. Observe that this lower bound for ${\lambda_\Omega}$ is strictly greater than $2$ because
$\|K\|_\sigma>0$.

We now prove $(ii)$. Assume that $\lambda>0$ is such that $\Ker(1/\lambda+C_\sigma^m)\neq0$. Let $0\neq g\in L^2(\sigma)^4$ be such that $C_\sigma^m(g)=-g/\lambda.$ In view of (\ref{sphere eq1}),
\begin{equation}\label{quad eq8}
\begin{split}
\text{if}\quad g=\left( \begin{array}{rr}  u  \\  h  \end{array} \right)
\quad\text{then}\quad\left\{
\begin{array}{rr}  2mK(u)+W(h)=-u/\lambda,  \\  W(u)=-h/\lambda.  \end{array}\right.
\end{split}
\end{equation}
From Lemma \ref{sphere l1}$(ii)$ and the last equality in (\ref{quad eq8}) we deduce that
\begin{equation}\label{quad eq9'}
u=-4((\upsigma \cdot N)W)^2(u)
=\frac{4}{\lambda}(\upsigma \cdot N)W(\upsigma \cdot N)(h),
\end{equation}
which plugged in the other equation in (\ref{quad eq8}) yields
\begin{equation}\label{quad eq9}
\begin{split}
\left(\frac{8m}{\lambda}K(\upsigma \cdot N)W(\upsigma \cdot N)+W
+\frac{4}{\lambda^2}(\upsigma \cdot N)W(\upsigma \cdot N)\right)(h)=0.
\end{split}
\end{equation}
Using Lemma \ref{sphere l1}, we may write
\begin{equation}\label{quad eq10}
\begin{split}
\frac{8m}{\lambda}&K(\upsigma \cdot N)W(\upsigma \cdot N)+W
+\frac{4}{\lambda^2}(\upsigma \cdot N)W(\upsigma \cdot N)\\
&=-\frac{8m}{\lambda}W(\upsigma \cdot N)K(\upsigma \cdot N)
+W(\upsigma \cdot N)(\upsigma \cdot N)
-\frac{16}{\lambda^2}(W(\upsigma \cdot N))^2(\upsigma \cdot N)W(\upsigma \cdot N)\\
&=W(\upsigma \cdot N)\left(-\frac{8m}{\lambda}K
+1-\frac{16}{\lambda^2}W^2\right)(\upsigma \cdot N).
\end{split}
\end{equation}
Since $W(\upsigma \cdot N)$ is invertible by Lemma \ref{sphere l1}$(ii)$, from (\ref{quad eq9}) and (\ref{quad eq10}) we get that
\begin{equation}\label{quad eq11}
\begin{split}
\left(-\frac{8m}{\lambda}K
+1-\frac{16}{\lambda^2}W^2\right)(f)=0,
\end{split}
\end{equation}
where we have set $f=(\upsigma \cdot N)h$. Note that $u$ is given in terms of $h$ by (\ref{quad eq9'}) and $g\neq0$ by assumption, thus we can also assume that $f\neq0$. In conclusion, we have seen that if $\Ker(1/\lambda+C_\sigma^m)\neq0$ then there exists $0\neq f\in L^2(\sigma)^2$ such that (\ref{quad eq11}) holds. Actually, since all the involved arguments are reversible, we see that
\begin{equation}\label{quad eq12}
\Ker(1/\lambda+C_\sigma^m)\neq0\quad\text{if and only if}\quad
\Ker(-(8m/\lambda)K+1-(16/\lambda^2)W^2)\neq 0.
\end{equation}
Moreover, if we multiply $(\ref{quad eq11})$ by $\overline f$ and we integrate with respect to $\sigma$, using the self-adjointness of $W$ we get
\begin{equation}\label{quad eq16}
A(\lambda,f)=\|f\|_\sigma^2\qquad\text{for all }f\in \Ker(-(8m/\lambda)K+1-(16/\lambda^2)W^2).
\end{equation}
Using that $W$ is invertible and that $K$ is positive, it is easy to show that
$A(\lambda-\epsilon,f)>A(\lambda,f)$ for all $f\neq0$ and all $0<\epsilon<\lambda$. In particular, $A(\lambda-\epsilon,f)>\|f\|_\sigma^2$ for all $0\neq f\in \Ker(-(8m/\lambda)K+1-(16/\lambda^2)W^2)$, which easily implies that $\lambda-\epsilon\leq{\lambda_\Omega}$ for all $0<\epsilon<\lambda$ whenever $\Ker(-(8m/\lambda)K+1-(16/\lambda^2)W^2)\neq0$. Finally, applying (\ref{quad eq12}) and taking $\epsilon\to0$ we conclude that if
$\Ker(1/\lambda+C_\sigma^m)\neq0$ then $\lambda\leq{\lambda_\Omega}$, and the proof of $(ii)$ is complete.

Concerning $(iii)$, if one repeats the arguments used to prove $(ii)$ but on the assumtion that $\Ker(1/\lambda+C_\sigma^{-m})\neq0$, one can show that there exists some $f\in L^2(\sigma)^2$ such that
\begin{equation*}
\begin{split}
0=\left(\frac{8m}{\lambda}K
+1-\frac{16}{\lambda^2}W^2\right)(f)
=\left(-\frac{8m}{|\lambda|}K
+1-\frac{16}{|\lambda|^2}W^2\right)(f),
\end{split}
\end{equation*}
since we are assuming $\lambda<0$. Hence, we are reduced to the case treated in (\ref{quad eq11}) but with the parameter $|\lambda|$. The rest of the proof follows the same lines, getting that $-\lambda=|\lambda|\leq{\lambda_\Omega}$. In particular, we also obtain that
\begin{equation}\label{quad eq12'}
\Ker(1/\lambda+C_\sigma^{-m})\neq0\quad\text{if and only if}\quad
\Ker((8m/\lambda)K+1-(16/\lambda^2)W^2)\neq 0.
\end{equation}

Let us prove $(iv)$. Since $K$ is positive, $A(\lambda,f)$ is a non-increasing function of $\lambda>0$ for all $f\in L^2(\sigma)^2$. By the definiton of ${\lambda_\Omega}$, this monotony implies that (\ref{quad eq1}) holds for all $\lambda\geq{\lambda_\Omega}$ and it is sharp for $\lambda={\lambda_\Omega}$. It remains to be shown that if $\lambda_\Omega>2\sqrt{2}$ then the equality is attained and that the minimizers  give rise to functions in $\Ker(1/{\lambda_\Omega}+C_\sigma^m)$ and vice versa. As we did in (\ref{quad isometry eq3}),
\begin{equation*}
\begin{split}
\int|W(f)|^2\,d\sigma
=\int W(\upsigma\cdot N)\{\upsigma\cdot N,W\}(f)\cdot\overline{f}\,d\sigma+\frac{1}{4}\int|f|^2\,d\sigma
\end{split}
\end{equation*}
for all $f\in L^2(\sigma)^2$. Set
$$T=W(\upsigma\cdot N)\{\upsigma\cdot N,W\}=W^2-\frac{1}{4}.$$
From Lemma \ref{sphere l1} we see that $T$ is self-adjoint and, since $\partial\Omega$ is $\CC^\infty$, it is also compact by the same arguments that prove \cite[Lemma 3.5]{AMV1}. Now, we can write
\begin{equation}\label{quad eq13}
\begin{split}
A(\lambda,f)-\|f\|_\sigma^2=
\int \left(\frac{16}{\lambda^2}\,T+\frac{8m}{\lambda}\,K\right)(f)\cdot\overline{f}\,d\sigma
+\left(\frac{4}{\lambda^2}-1\right)\int|f|^2\,d\sigma.
\end{split}
\end{equation}
Let $\lambda>2$. Then, (\ref{quad eq13}) shows that
\begin{equation}\label{quad eq14}
\begin{split}
A(\lambda,f)\leq\|f\|_\sigma^2\quad \text{if and only if}\quad
\int T_\lambda(f)\cdot\overline{f}\,d\sigma
\leq\int|f|^2\,d\sigma,
\end{split}
\end{equation}
where we have set $$T_\lambda=\frac{4}{\lambda^2-4}\left(4T+2m\lambda K\right),$$
and the same holds replacing ``$\leq$'' by ``$=$'' or ``$>$'' on both sides of (\ref{quad eq14}). Observe that $T_\lambda$ is also self-adjoint and compact (for all real $\lambda\neq\pm2$), since $T$ and $K$ also are. In particular, by \cite[Lemma (0.43)]{Folland}, there exists $0\neq f_\lambda\in L^2(\sigma)^2$ such that
$T_\lambda(f_\lambda)=\|T_\lambda\|_\sigma f_\lambda$ or $T_\lambda(f_\lambda)=-\|T_\lambda\|_\sigma f_\lambda$.

We are going to show that if $\lambda_\Omega>2\sqrt{2}$  then we must have $T_{\lambda_\Omega}(f_{\lambda_\Omega})=\|T_{\lambda_\Omega}\|_\sigma f_{\lambda_\Omega}$ with $\|T_{\lambda_\Omega}\|_\sigma=1$. Using (\ref{quad eq13}) we see that, if $\lambda>2$,
\begin{equation}\label{quad eq14'}
\begin{split}
A(\lambda,f)\geq 0\quad \text{if and only if}\quad
\int T_\lambda(f)\cdot\overline{f}\,d\sigma
\geq-\frac{4}{\lambda^2-4}\int|f|^2\,d\sigma.
\end{split}
\end{equation}
Since $A(\lambda,f)\geq 0$ for all $\lambda>0$ because $K$ is positive, from (\ref{quad eq14'}) we get that
\begin{equation}\label{quad eq14''}
\begin{split}
\int T_\lambda(f)\cdot\overline{f}\,d\sigma
>-\int|f|^2\,d\sigma \qquad \text{for all }0\neq f\in L^2(\sigma)^2\text{ and all }\lambda>2\sqrt{2}.
\end{split}
\end{equation}
Combining (\ref{quad eq14}) and (\ref{quad eq14''}), and using that $T_\lambda$ is self-adjoint, we deduce that
\begin{equation}\label{quad eq14'''}
\begin{split}
\|T_\lambda\|_\sigma=\sup_{\|f\|_\sigma=1}\left|\int T_\lambda(f)\cdot\overline{f}\,d\sigma\right|
\leq1 \qquad \text{for all }\lambda\geq\lambda_\Omega\text{ if }\lambda_\Omega>2\sqrt{2}.
\end{split}
\end{equation}
Furthermore, the definition of $\lambda_\Omega$ and (\ref{quad eq14}) imply that $\|T_\lambda\|_\sigma>1$ for all $\lambda<\lambda_\Omega$. From this and (\ref{quad eq14'''}) we get that if $\lambda_\Omega>2\sqrt{2}$ then $\|T_{\lambda_\Omega}\|_\sigma=1$, since $\|T_\lambda\|_\sigma$ depends continuously on $\lambda$ for all $\lambda>2$. In particular, we have seen that there exists $0\neq f_{\lambda_\Omega}\in L^2(\sigma)^2$ such that
$T_{\lambda_\Omega}(f_{\lambda_\Omega})= f_{\lambda_\Omega}$ or $T_{\lambda_\Omega}(f_{\lambda_\Omega})=- f_{\lambda_\Omega}$. However, if $\lambda_\Omega>2\sqrt{2}$ then (\ref{quad eq14''}) shows that the case $T_{\lambda_\Omega}(f_{\lambda_\Omega})=- f_{\lambda_\Omega}$ is not possible, thus $T_{\lambda_\Omega}(f_{\lambda_\Omega})= f_{\lambda_\Omega}$ as claimed.
From (\ref{quad eq14}), we finally get $A({\lambda_\Omega},f_{\lambda_\Omega})=\|f_{\lambda_\Omega}\|_\sigma^2$, which proves that the equality in $(\ref{quad eq1})$ is attained for $\lambda={\lambda_\Omega}$ on $f_{\lambda_\Omega}$.

Concerning the minimizers of (\ref{quad eq1}), assume that $A({\lambda_\Omega},f)=\|f\|_\sigma^2$ for some $f\neq0$. Then (\ref{quad eq14}) gives
\begin{equation}\label{quad eq15}
\begin{split}
\int T_{{\lambda_\Omega}}(f)\cdot\overline{f}\,d\sigma
=\int|f|^2\,d\sigma.
\end{split}
\end{equation}
Since $T_{{\lambda_\Omega}}$ is a compact self-ajoint operator (so it diagonalizes in an orthonormal basis of eigenvectors) and $\|T_{{\lambda_\Omega}}\|_\sigma=1$ if $\lambda_\Omega>2\sqrt{2}$, by (\ref{quad eq15}) we must have
$T_{{\lambda_\Omega}}(f)= f$. From the definitions of $T_{\lambda_\Omega}$ and $T$, we get
\begin{equation*}
\left(16W^2+8m{\lambda_\Omega} K-{{\lambda_\Omega}^2}\right)(f)=0,
\end{equation*}
which implies that $\Ker(1/{\lambda_\Omega}+C_\sigma^m)\neq0$ by (\ref{quad eq12}) and that
$\Ker(-1/{\lambda_\Omega}+C_\sigma^{-m})\neq0$ by (\ref{quad eq12'}). On the contrary, if $\Ker(1/{\lambda_\Omega}+C_\sigma^m)\neq0$ then (\ref{quad eq12}) and (\ref{quad eq16}) show that there exists $f\neq0$ such that $A({\lambda_\Omega},f)=\|f\|_\sigma^2$, and similarly for $\Ker(-1/{\lambda_\Omega}+C_\sigma^{-m})\neq0$ using (\ref{quad eq12'}) and (\ref{quad eq16}).

Regarding $(v)$, one can check that the conclusions in $(iv)$ also hold when one works with (\ref{quad eq1''}) instead of (\ref{quad eq1}) since these quadratic form inequalities are equivalent (recall the computations carried out between (\ref{quad eq3}) and (\ref{quad eq7})), we leave the details for the reader. The theorem is finally proved.
\end{proof}

\begin{remark}
Gathering (\ref{quad eq12}) and (\ref{quad eq12'}), we get that
\begin{equation}\label{quad rem1 eq1}
\Ker(1/\lambda+C_\sigma^{m})\neq0\quad\text{if and only if}\quad
\Ker(-1/\lambda+C_\sigma^{-m})\neq0.
\end{equation}
This corresponds to the endpoint case $a=m$ in \cite[Theorem 3.6]{AMV2}, thanks to Proposition \ref{spec p1}. The relevant fact here is that, despite that in \cite[Theorem 3.6]{AMV2} we were assuming some invariance of $\sigma$ with respect to reflections in order to obtain the antisymmetry property of the eigenvalues with respect to the potential, (\ref{quad rem1 eq1}) holds without this assumption on $\sigma$.
\end{remark}

\begin{remark}\label{constraint remark}
The assumption $\lambda_\Omega>2\sqrt{2}$ in Theorem \ref{quad t1}$(iv)$ can be weakened using essentially the same arguments as before. Roughly speaking, from (\ref{quad isometry eq1}) one sees that 
$16\|W\|^2_\sigma W^2\geq1$, considering this as an inequality between operators in the sense of quadratic forms. Then
\begin{equation}\label{weakening}
\begin{split}
T_\lambda&=\frac{4\left(4T+2m\lambda K\right)}{\lambda^2-4}
\geq\frac{16T}{\lambda^2-4}=\frac{16W^2-4}{\lambda^2-4}
\geq\frac{1}{\lambda^2-4}\left(\frac{1}{\|W\|^2_\sigma}-4\right).
\end{split}
\end{equation}
The right hand side of (\ref{quad eq14'}) formally corresponds to the limiting case $\|W\|_\sigma=\infty$ in (\ref{weakening}). Since the arguments in the proof of Theorem \ref{quad t1}$(iv)$ require that $T_{\lambda_\Omega}>-1$ in order to get $\|T_{\lambda_\Omega}\|_\sigma=1$ and find the minimizers, in view of (\ref{weakening}) one sees that a possible assumption is 
\begin{equation}\label{weakening1}
\lambda_\Omega>2\sqrt{2-\frac{1}{4\|W\|_\sigma^2}},
\end{equation} 
which is weaker than $\lambda_\Omega>2\sqrt{2}$. Following the arguments in the forthcoming pages of this article until the proof of Theorem \ref{intro thm 1} but using (\ref{weakening1}) instead of $\lambda_\Omega>2\sqrt{2}$, one can see that (\ref{constr q}) can be weakened to 
\begin{equation*}
m\,\frac{\Area(\partial\Omega)}{\Capa(\Omega)}
>\frac{1-\frac{1}{4\|W\|_\sigma^2}}
{4\sqrt{2-\frac{1}{4\|W\|_\sigma^2}}}.
\end{equation*}
However, in what respects to the potential applications of Theorem \ref{intro thm 1} as an isoperimetric-type inequality, one may find bounded domains $\Omega$ with constant $m\Area(\partial\Omega)/\Capa(\Omega)$ but with $\|W\|_\sigma$ arbitrarily large, since this last quantity strongly depends on the abruptness of $\partial\Omega$. As a consequence, in general one has to assume (\ref{constr q}) (or equivalently the limiting case $\lambda_\Omega>2\sqrt{2}$) to make use of Theorem \ref{intro thm 1}. 

\end{remark}

\begin{corollary}\label{quad c1}
Let $\lambda_{ m}^s$ and $\lambda_{- m}^i$ be as in {\em Corollary \ref{elec c1}}. If $\lambda_\Omega>2\sqrt{2}$, then
\begin{equation*}
\begin{split}
{\lambda_\Omega}&=\lambda_{m}^s=-\lambda_{-m}^i\\
&=\sup\{\lambda\in\R:\,\Ker(H+V_{\lambda}-a)\neq0\text{ for some }a\in(-m,m)\}\\
&=-\inf\{\lambda\in\R:\,\Ker(H+V_{\lambda}-a)\neq0\text{ for some }a\in(-m,m)\}\\
&=\sup\{|\lambda|:\,\Ker(H+V_{\lambda}-a)\neq0\text{ for some }a\in(-m,m)\}
\end{split}
\end{equation*}
and $4/{\lambda_\Omega}=\inf\{|\lambda|:\,\Ker(H+V_{\lambda}-a)\neq0\text{ for some }a\in(-m,m)\}$.
\end{corollary}

\begin{proof}
From Corollary \ref{elec c1} we already know that $\lambda_{m}^s>0>\lambda_{-m}^i$. Combining Theorem \ref{quad t1}$(ii)$ and $(iv)$, we easily see that
$$\lambda_{m}^s=\sup\{\lambda\in\R:\,\Ker(1/\lambda+C^{m}_\sigma)\neq0\}={\lambda_\Omega},$$
and similarly, using Theorem \ref{quad t1}$(iii)$ and $(iv)$, we get that $\lambda_{-m}^i=-{\lambda_\Omega}$. Hence, the corollary follows directly from Corollary \ref{elec c1}. Observe that the supremum and the infimum in the definitions of $\lambda_{m}^s$ and $\lambda_{-m}^i$ are a maximum and a minimum, respectively, if $\lambda_\Omega>2\sqrt{2}$.
\end{proof}

\begin{remark}\label{rr1}
Combining the methods used above one can also show that, if $\lambda_\Omega>2\sqrt{2}$,
\begin{equation*}
\begin{split}
4/{\lambda_\Omega}&=\inf\{\lambda>0:\,\Ker(H+V_{\lambda}-a)\neq0\text{ for some }a\in(-m,m)\}\\
&=-\sup\{\lambda<0:\,\Ker(H+V_{\lambda}-a)\neq0\text{ for some }a\in(-m,m)\}.
\end{split}
\end{equation*}
\end{remark}

\section{An isoperimetric-type inequality}\label{ss isoper}
For the sake of clarity, given a bounded open set  $\Omega\subset\R^3$ with smooth boundary, we set $$\Volume(\Omega)=\mu(\Omega)
\quad\text{and}\quad
\Area(\partial\Omega)=\sigma(\partial\Omega).$$
Furthermore, to stress the dependence of $K$ and $W$ on $\sigma$ (that is, on $\partial\Omega$), we write $K_\Omega$ and $W_\Omega$ respectively.

\subsection{A test to exclude constraints on $\Omega$}
In the setting of bounded domains with smooth boundary, due to Theorem \ref{quad t1}$(i)$ we have
$$4\big(m\|K_\Omega\|_\sigma
+\sqrt{m^2\|K_\Omega\|_\sigma^2+1/4}\big)\leq\lambda_\Omega\leq
4\big(m\|K_\Omega\|_\sigma
+\sqrt{m^2\|K_\Omega\|_\sigma^2+\|W_\Omega\|_\sigma^2}\big).$$
Since $\|W_\Omega\|_\sigma^2=1/4$ if and only if $\partial\Omega$ is a sphere (recall \cite{HMMPT}), one may be tempted to look for an isoperimetric-type inequality for $\|K_\Omega\|_\sigma$ so that the ball is a minimizer, and thus obtaining an inequality for $\lambda_\Omega$. In order to do so, one may impose some constraint on the admissible domains because of the rescaling properties of $\|K_\Omega\|_\sigma$ under dilations; if $$\Omega_t=\{tx:\,x\in\Omega\}\quad\text{for }t>0$$ and $\sigma_t$ is the surface measure on $\partial\Omega_t$ then $\|K_{\Omega_t}\|_{\sigma_t}=O(t)$ but
$\|W_{\Omega_t}\|_{\sigma_t}=\|W_{\Omega}\|_{\sigma}=O(1)$.
We are going to present a simple and classical method to test possible constraints that do not permit the existence of domains that minimize $\|K_\Omega\|_\sigma$. Roughly speaking, the method is based on the splitting of a domain into two suitable copies of itself. In particular, it allows us to prove that ``{\em there is no bounded domain with smooth boundary that attains the infimum of $\|K_\Omega\|_\sigma$ over all bounded domains  $\Omega$ with smooth boundary and constant volume}''. The same holds replacing ``volume'' by ``area of the boundary''.

For $t>0$ and $z\in\R^3$, we set
$\Omega_{t,z}=\Omega_t\cup(\Omega_t+z)$ and we denote by $\sigma_{t,z}$ the surface measure on $\partial\Omega_{t,z}$.
We assume that $|z|$ is big enough, so  $\Omega_t\cap(\Omega_t+z)=\emptyset$.

\begin{lemma}\label{iso l2}
Given $\Omega\subset\R^3$ and $t>0$, if $|z|$ is big enough then
\begin{equation*}
\begin{split}
\left|\|K_{\Omega_{t,z}}\|_{\sigma_{t,z}}-t\|K_{\Omega}\|_{\sigma}\right|
\leq\frac{\sigma(\partial\Omega_{t,z})}{2\pi\dist(\partial\Omega_t,\partial\Omega_t+z)}.
\end{split}
\end{equation*}
\end{lemma}

\begin{proof}
Since $K_{\Omega_t}$ is positive and self-adjoint, a change of variables easily yields
\begin{equation}\label{iso eq5}
\begin{split}
\|K_{\Omega_t}\|_{\sigma_t}
&=\sup_{f\neq0}\frac{1}{\|f\|_{\sigma_t}^2}\int K_{\Omega_t}(f)\cdot \overline{f}\,d\sigma_t
=\sup_{f\neq0}\frac{1}{\|f\|_{\sigma_t}^2}
\iint \frac{f(x)\cdot\overline{f(y)}}{4\pi|x-y|}\,d\sigma_t(x)\,d\sigma_t(y)\\
&=\sup_{f\neq0}\frac{t}{\int|f(tx)|^2\,d\sigma(x)}
\iint \frac{f(tx)\cdot\overline{f(ty)}}{4\pi|x-y|}\,d\sigma(x)\,d\sigma(y)=t\|K_{\Omega}\|_{\sigma}.
\end{split}
\end{equation}

Given $f\in L^2(\sigma_{t,z})^2$, set
\begin{equation*}
I(f)=\iint \frac{f(x)\cdot\overline{f(y)}}{4\pi|x-y|}\,d\sigma_{t,z}(x)\,d\sigma_{t,z}(y).
\end{equation*}
Since $\partial\Omega_{t,z}=\partial\Omega_t\cup(\partial\Omega_t+z)$, using Fubini's theorem we can decompose
\begin{equation}\label{iso eq10}
\begin{split}
I(f)\!&=\!
\left(\iint_{\partial\Omega_t\times\partial\Omega_t} \!\!\!\!
+\iint_{(\partial\Omega_t+z)\times(\partial\Omega_t+z)} \!\!\!\!
+2\real\iint_{\partial\Omega_t\times(\partial\Omega_t+z)}\right)\! \frac{f(x)\cdot\overline{f(y)}}{4\pi|x-y|}\,d\sigma_{t,z}(x)\,d\sigma_{t,z}(y)
\\&=:I_1(f)+I_2(f)+I_3(f).
\end{split}
\end{equation}
Note that $\|K_{\Omega_t+z}\|_{\sigma_t+z}=\|K_{\Omega_t}\|_{\sigma_t}$ and
$\|f\|_{\sigma_{t}}^2+\|f\|_{\sigma_{t}+z}^2
=\|f\|_{\sigma_{t,z}}^2$, thus using (\ref{iso eq5}) we get
\begin{equation}\label{iso eq6}
\begin{split}
I_1(f)+I_2(f)
&=\|f\|_{\sigma_{t}}^2\frac{I_1(f)}{\|f\|_{\sigma_{t}}^2}
+\|f\|_{\sigma_{t}+z}^2\frac{I_2(f)}{\|f\|_{\sigma_{t}+z}^2}\\
&\leq \|f\|_{\sigma_{t}}^2\|K_{\Omega_t}\|_{\sigma_t}
+\|f\|_{\sigma_{t}+z}^2\|K_{\Omega_t+z}\|_{\sigma_t+z}
=\|f\|_{\sigma_{t,z}}^2 t\|K_{\Omega}\|_{\sigma}.
\end{split}
\end{equation}
Moreover, by H\"older's inequality,
\begin{equation}\label{iso eq11}
\begin{split}
|I_3(f)|\leq\frac{\int_{\partial\Omega_t}|f|\,d\sigma_{t,z}
\int_{\partial\Omega_t+z}|f|\,d\sigma_{t,z}}{2\pi\dist(\partial\Omega_t,\partial\Omega_t+z)}
\leq\frac{\sigma(\partial\Omega_{t,z})\|f\|^2_{\sigma_{t,z}}}{2\pi\dist(\partial\Omega_t,\partial\Omega_t+z)}.
\end{split}
\end{equation}
Dividing by $\|f\|_{\sigma_{t,z}}^2$ and taking the supremum over all $f\neq0$ in (\ref{iso eq10}), and using (\ref{iso eq6}) and (\ref{iso eq11}), we finally obtain
\begin{equation}\label{iso eq8}
\begin{split}
\|K_{\Omega_{t,z}}\|_{\sigma_{t,z}}
&=\sup_{f\neq0}\frac{I(f)}{\|f\|^2_{\sigma_{t,z}}}
\leq\sup_{f\neq0}\frac{I_1(f)+I_2(f)}{\|f\|^2_{\sigma_{t,z}}}+
\sup_{f\neq0}\frac{I_3(f)}{\|f\|^2_{\sigma_{t,z}}}\\
&\leq t\|K_{\Omega}\|_{\sigma}
+\frac{\sigma(\partial\Omega_{t,z})}{2\pi\dist(\partial\Omega_t,\partial\Omega_t+z)}.
\end{split}
\end{equation}

Given $f\in L^2(\sigma_t)^2$ and $g\in L^2(\sigma_t+z)^2$, set
$h=f\chi_{\partial\Omega_t}/\|f\|_{\sigma_t}
+g\chi_{\partial\Omega_t+z}/\|g\|_{\sigma_t+z}\in L^2(\sigma_{t,z})^2$ (we extended $f$ to be identically zero in $\partial\Omega_t+z$ and analogously for $g$). Then $\|h\|_{\sigma_{t,z}}^2=2$ and, using (\ref{iso eq10}) and (\ref{iso eq11}) on $h$, we get
\begin{equation}\label{iso eq7}
\begin{split}
\frac{I_1(f)}{\|f\|_{\sigma_t}^2}+\frac{I_2(g)}{\|g\|_{\sigma_t+z}^2}
&=I_1(h)+I_2(h)
\leq 2\,\frac{I(h)}{\|h\|_{\sigma_{t,z}}^2}+|I_3(h)|\\
&\leq2\|K_{\Omega_{t,z}}\|_{\sigma_{t,z}}+\frac{\sigma(\partial\Omega_{t,z})}{\pi\dist(\partial\Omega_t,\partial\Omega_t+z)}.
\end{split}
\end{equation}
Taking the supremum over all $f\neq0$ and $g\neq0$ in (\ref{iso eq7}), and since $K_{\Omega_t}$ and $K_{\Omega_t+z}$ are positive and self-adjoint, using (\ref{iso eq5}) we see that
\begin{equation}\label{iso eq9}
\begin{split}
2t\|K_{\Omega}\|_{\sigma}=\|K_{\Omega_t}\|_{\sigma_t}+\|K_{\Omega_t+z}\|_{\sigma_t+z}
\leq2\|K_{\Omega_{t,z}}\|_{\sigma_{t,z}}+\frac{\sigma(\partial\Omega_{t,z})}{\pi\dist(\partial\Omega_t,\partial\Omega_t+z)}.
\end{split}
\end{equation}
The lemma follows from (\ref{iso eq8}) and (\ref{iso eq9}).
\end{proof}

With Lemma \ref{iso l2} at our disposal, we can easily prove that ``{\em there is no bounded open set with smooth boundary that attains the infimum of $\|K_\Omega\|_\sigma$ over all bounded open sets $\Omega$ with smooth boundary and constant volume}'', and that the same holds replacing ``volume'' by ``area of the boundary''. Let $\Omega$ be a bounded open set with smooth boundary. If $|z|$ is big enough,
$\Volume(\Omega_{2^{-1/3},z})
=2\Volume(\Omega_{2^{-1/3}})=\Volume(\Omega)$,
and Lemma \ref{iso l2} shows that
\begin{equation*}
\begin{split}
\|K_{\Omega_{2^{-1/3},z}}\|_{\sigma_{2^{-1/3},z}}
\leq 2^{-1/3}\|K_{\Omega}\|_{\sigma}
+\frac{\sigma(\partial\Omega_{2^{-1/3}})}{\pi\dist(\partial\Omega_{2^{-1/3}},\partial\Omega_{2^{-1/3}}+z)}
<\|K_{\Omega}\|_{\sigma},
\end{split}
\end{equation*}
thus given $\Omega$ we have constructed another bounded domain  $\Omega_{2^{-1/3},z}$ with smooth boundary, with the same volume as $\Omega$, but with a strictly smaller norm of the associated operator $K$. Hence, there can not exists a minimizer. In case that the constraint concerns ``constant area of the boundary'', one only needs to argue with $\Omega_{2^{-1/2},z}$ instead of $\Omega_{2^{-1/3},z}$.

Finally, under the assumption of connectedness, the statement ``{\em there is no bounded domain with smooth boundary that attains the infimum of $\|K_\Omega\|_\sigma$ over all bounded domains  $\Omega$ with smooth boundary and constant volume}'' can be proven with the same arguments as before but connecting, in a smooth way, the two connected components of $\Omega_{t,z}$ (once $t$ and $z$ are properly chosen) by a thin tube and showing that the contribution of the tube in $\|K_{\Omega_{t,z}}\|_{\sigma_{t,z}}$ is as small as we want by taking the tube thin enough, essentially because the kernel $k$ is locally integrable with respect to surface measure. We leave the details for the reader.

\subsection{The relation with the Newtonian capacity}\label{ss Newton}
Given a compact set $E\subset\R^3$, the {\em Newtonian capacity} of $E$ (sometimes referred in the literature as {\em electrostatic} or {\em harmonic capacity}) is defined by
\begin{equation*}
\Capa(E)=\left(\inf_\nu\iint\frac{d\nu(x)\,d\nu(y)}{4\pi|x-y|}\right)^{-1},
\end{equation*}
where the infimum is taken over all probability Borel measures $\nu$ supported in $E$. Sometimes in the literature, the $4\pi$ appearing in the definition of $\Capa(E)$ is changed by another precise constant. For the case of open sets $U\subset\R^3$, one defines
$$\Capa(U)
=\sup\{\Capa(E):\,E\subset U,\, E\text{ compact}\}.$$
The Newtonian capacity has a number of distinguished properties which we state as a lemma for future applications
(see \cite[Chapters 9 and 11]{LL} or \cite{PS}, for example).
\begin{lemma}\label{iso l1}
Let $\Omega$ be a bounded open set with smooth boundary. Then,
\begin{itemize}
\item[$(i)$] $\Capa(\Omega)=\Capa(\overline{\Omega})=\Capa(\partial\Omega)$,
\item[$(ii)$] P\'olya-Szeg\"o inequality: let $\Omega^*$ be the closed ball centered at the origin such that $\Volume(\Omega^*)=\Volume(\Omega)$. Then
$\Capa(\Omega)\geq\Capa(\Omega^*).$ Moreover, the equality holds if and only if $\Omega$ is a ball.
\item[$(iii)$] $\Capa(\Omega)=2(6\pi^2)^{1/3}\Volume(\Omega)^{1/3}$ if $\Omega$ is a ball.
\end{itemize}
\end{lemma}
Regarding the uniqueness of the minimizer in Lemma \ref{iso l1}$(ii)$, it is important to impose some restriction on $\partial\Omega$ (such as  regularity) in order to avoid sets of Newtonian capacity zero.

\begin{lemma}\label{iso l3}
Let $\Omega\subset\R^3$ be a bounded domain with smooth boundary. Then
\begin{equation}\label{iso eq1}
\lambda_{\Omega}\geq
4\Bigg(m\,\frac{\Area(\partial\Omega)}{\Capa(\Omega)}+\sqrt{m^2\left(\frac{\Area(\partial\Omega)}{\Capa(\Omega)}\right)^2
+\frac{1}{4}}\Bigg),
\end{equation}
and the equality holds if and only if $\Omega$ is a ball.
\end{lemma}

\begin{proof}
Since $K_\Omega$ is a positive self-adjoint operator, we have
\begin{equation}\label{iso eq3}
\begin{split}
\|K_\Omega\|_\sigma
&=\sup_{f\neq0}\frac{1}{\|f\|_\sigma^2}\int K_\Omega(f)\cdot \overline{f}\,d\sigma
\geq\frac{1}{\sigma(\partial\Omega)}\int K_\Omega\left(\!\!\left(\begin{array}{c} 1\\
0 \end{array}\right)\!\!\right)\cdot \overline{\left(\begin{array}{c} 1\\
0 \end{array}\right)}\,d\sigma\\
&=\sigma(\partial\Omega)\iint \frac{1}{4\pi|x-y|}
\,\frac{d\sigma(x)}{\sigma(\partial\Omega)}
\,\frac{d\sigma(y)}{\sigma(\partial\Omega)}
\geq\frac{\Area(\partial\Omega)}{\Capa(\Omega)},
\end{split}
\end{equation}
where we also used Lemma \ref{iso l1}$(i)$ in the last inequality above.
Gathering (\ref{iso eq3}) and Theorem \ref{quad t1}$(i)$, we get (\ref{iso eq1}).

Assume that $\Omega$ is a ball of radius $r>0$ centered at the origin. Then, for any $x,y\in\partial\Omega$,
\begin{equation*}
\begin{split}
(\upsigma\cdot N(x))(\upsigma\cdot (x-y))
&=\frac{1}{r}(\upsigma\cdot x)(\upsigma\cdot (x-y))
=\frac{1}{r}(r^2-(\upsigma\cdot x)(\upsigma\cdot y))\\
&=-\frac{1}{r}(-r^2+(\upsigma\cdot x)(\upsigma\cdot y))
=-\frac{1}{r}(\upsigma\cdot (x-y))(\upsigma\cdot y))\\
&=-(\upsigma\cdot (x-y))(\upsigma\cdot N(y)).
\end{split}
\end{equation*}
This identity easily yields $\{\upsigma\cdot N,W_\Omega\}=0$ and, by Lemma \ref{quad isometry},  $\|W_\Omega\|_\sigma=1/2$. Therefore, from Theorem \ref{quad t1}$(i)$ we get that
\begin{equation}\label{iso eq4}
\begin{split}
\lambda_\Omega=4\big(m\|K_\Omega\|_\sigma
+\sqrt{m^2\|K_\Omega\|_\sigma^2+1/4}\big)
\end{split}
\end{equation}
if $\Omega$ is a ball.
Let $e_3=(0,0,1)\in\R^3$, and we identify the matrix $k$ with its scalar version. Following \cite[Generalized Young's Inequality (0.10)]{Folland} and since $\partial\Omega$ is invariant under rotations, it is easy to see that
\begin{equation*}
\begin{split}
\|K_\Omega\|_{\sigma}\leq\|k(\cdot-re_3)\|_{L^1(\sigma)}
=\frac{1}{\sigma(\partial\Omega)}\iint\frac{d\sigma(y)\,d\sigma(x)}{4\pi|re_3-y|}
=\frac{1}{\sigma(\partial\Omega)}\iint\frac{d\sigma(y)\,d\sigma(x)}{4\pi|x-y|}.
\end{split}
\end{equation*}
In particular, this shows that the first inequality in (\ref{iso eq3}) is an equality if $\Omega$ is a ball. It is well-known that the infimum in the definition of $\Capa(\Omega)$ is attained on the normalized surface measure $\sigma/\sigma(\partial\Omega)$ when $\Omega$ is a ball (the Newtonian potential of $\sigma/\sigma(\partial\Omega)$ corresponds to the harmonic function that takes the constant value 1 on $\partial\Omega$, zero at infinity and minimizes the exterior Dirichlet energy), thus the second inequality in (\ref{iso eq3}) is also an equality in this case. Therefore,
$\|K_\Omega\|_\sigma=\Area(\partial\Omega)/\Capa(\Omega)$ if $\Omega$ is a ball, which combined with (\ref{iso eq4}) proves that (\ref{iso eq1}) is an equality in this case.

On the contrary, assume that (\ref{iso eq1}) is an equality. From (\ref{iso eq3}) and Theorem \ref{quad t1}$(i)$ we see that the second inequality in (\ref{iso eq3}) must be an equality, which means that the infimum in the definition of $\Capa(\Omega)$ is attained on $\sigma/\sigma(\partial\Omega)$. In the literature, the probability measure that gives the minimum in $\Capa(\Omega)$ is referred as {\em equilibrium distribution}. Let us recall Gruber's conjecture (see \cite[Section 4.1]{HMMPT}): ``the equilibrium distribution of $\overline\Omega$ is $c\sigma$ for some $c>0$ if and only if $\Omega$ is a ball''. In \cite{Reichel1} and \cite{Reichel2}, the author shows that Gruber's conjecture holds in the case of $\CC^{2,\epsilon}$-domains. Putting all together, we see that if $\Omega$ is a bounded and smooth domain such that the equality in (\ref{iso eq1}) holds, the equilibrium distribution of $\overline\Omega$ is $\sigma/\sigma(\partial\Omega)$, which implies that $\Omega$ is a ball by Gruber's conjecture. The lemma is finally proved.
\end{proof}

Despite $(\ref{iso eq1})$ is sharp, it may not be a completely satisfactory inequality in the sense that the right hand side involves some ``obscure'' term, namely $\Capa(\Omega)$, from a measure theoretic point of view. It would be interesting to derive some related inequality such that the right hand side only involves $\Area(\partial\Omega)$ and/or $\Volume(\Omega)$. This is precisely the purpose of the following corollary, where and isoperimetric-type inequality for the product $\lambda_{\Omega}\Capa(\Omega)$ is derived.

\begin{corollary}\label{iso c1}
Let $\Omega\subset\R^3$ be a bounded open set with smooth boundary. Then
\begin{equation*}
\lambda_{\Omega}\Capa(\Omega)\geq
4\left(m\Area(\partial\Omega)+\sqrt{m^2\Area(\partial\Omega)^2
+{6^{2/3}\pi^{4/3}}
\Volume(\Omega)^{2/3}}\right),
\end{equation*}
and the equality holds if and only if $\Omega$ is a ball.
\end{corollary}
\begin{proof}
From Lemma \ref{iso l1}$(ii)$ and $(iii)$, we get
\begin{equation*}
2(6\pi^2)^{1/3}\Volume(\Omega)^{1/3}
=2(6\pi^2)^{1/3}\Volume(\Omega^*)^{1/3}
=\Capa(\Omega^*)\leq\Capa(\Omega),
\end{equation*}
and the equality holds if and only if $\Omega$ is a ball.
The corollary follows from this and (\ref{iso eq1}).
\end{proof}

\begin{proof}[Proof of {\em Theorem \ref{intro thm 1}}]
This is a straightforward application of Corollary \ref{quad c1} and Lemma \ref{iso l3}, just observe that if
$4\sqrt{2}m\Area(\partial\Omega)/\Capa(\Omega)>1$ then $\lambda_\Omega>2\sqrt{2}$.
\end{proof}

\begin{remark}\label{rr2}
Combining Remark \ref{rr1} with Corollary \ref{quad c1} and Lemma \ref{iso l3}, we see that Theorem \ref{intro thm 1} also holds replacing
\begin{center}
$\sup\{|\lambda|:\ldots\}$ by $\sup\{\lambda>0:\ldots\}$ or $-\inf\{|\lambda<0|:\ldots\}$, and\\
$\inf\{|\lambda|:\ldots\}$ by $\inf\{\lambda>0:\ldots\}$ or $-\sup\{\lambda<0:\ldots\}$.
\end{center}
\end{remark}

\end{document}